\documentclass[journal,onecolumn]{IEEEtran}
\usepackage{amsmath,amsfonts}
\usepackage{amssymb}
\usepackage{array}
\usepackage[caption=false,font=normalsize,labelfont=sf,textfont=sf]{subfig}
\usepackage{textcomp}
\usepackage{stfloats}
\usepackage{url}
\usepackage{color}
\usepackage{verbatim}
\usepackage{graphicx}
\usepackage{float}

\usepackage{amsmath}
\usepackage{amssymb}
\usepackage{mathrsfs}
\usepackage{makecell}
\usepackage{threeparttable}
\usepackage[figuresright]{rotating}
\usepackage{soul} 
\usepackage{color, xcolor} 
\usepackage{booktabs}
\usepackage{bm}
\usepackage{multicol}
\usepackage{algorithm}      
\usepackage{algpseudocode}
\usepackage{cite}

\usepackage{color}
\usepackage{colortbl}

\usepackage{setspace} 
\makeatletter
\let\NAT@parse\undefined
\makeatother
\usepackage{hyperref}  
\usepackage{diagbox}

\hypersetup{
	colorlinks=true,
	linkcolor=black,
	citecolor=black
}

\usepackage{tikz}
\usetikzlibrary{shapes,arrows}

\newtheorem{lemma}{\textbf{Lemma}}

\newtheorem{theorem}{\textbf{Theorem}}

\newtheorem{definition}{\textbf{Definition}}
\newenvironment{proof}{{\noindent\it \textbf{Proof}:} }{\hfill $\square$\par}
\newtheorem{example}{\textbf{Example}}

\begin{document}
	
	\title{
		 Generalized Repetition Codes and Their Application to HARQ
	}
	\author{Chaofeng Guan, Gaojun Luo, Lan Luo, Yangyang Fei, Hong Wang ~\IEEEmembership{}
	}

	\markboth{}%
	{}
	
	\IEEEpubid{}

	\maketitle
	
	\begin{abstract}

The inherent uncertainty of communication channels implies that any coding scheme has a non-zero probability of failing to correct errors, making retransmission mechanisms essential. To ensure message reliability and integrity, a dual-layer redundancy framework is typically employed: error correction codes mitigate noise-induced impairments at the physical layer, while cyclic redundancy checks verify message integrity after decoding. Retransmission is initiated if verification fails. This operational model can be categorized into two types of repeated communication models: Type-I systems repeatedly transmit identical codewords, whereas Type-II systems transmit distinct coded representations of the same message.
The core challenge lies in maximizing the probability of correct message decoding within a limited number of transmission rounds through verification-based feedback mechanisms.

In this paper, we consider a scenario where the same error-correcting code is used for repeated transmissions, and we specifically propose two classes of generalized repetition codes (GRCs), corresponding to the two repeated communication models.
In contrast to classical theory, we regard GRCs as error-correcting codes under multiple metrics—that is, GRCs possess multiple minimum distances. This design enables GRCs to perform multi-round error correction under different metrics, achieving stronger error-correction capabilities than classical error-correcting codes.
However, the special structure of GRCs makes their construction more challenging, as it requires simultaneously optimizing multiple minimum distances. To address this, we separately investigate the bounds and constructions for Type-I and Type-II GRCs, and obtain numerous optimal Type-I and Type-II GRCs.
When GRCs are used for repeated communication, it is only necessary to perform permutation operations on the codewords or linear transformations on the messages during each retransmission. 
This endows GRCs with the potential to seamlessly integrate with existing Hybrid Automatic Repeat-reQuest (HARQ) protocols and storage systems.

	\end{abstract}
	
	\section{Introduction}

Reliability and low latency are fundamental requirements in modern communication systems. The inherent noise and interference in real-world communication channels make retransmission inevitable when channel impairments exceed the error correction capacity of the code. Reducing retransmission attempts is critical, as these processes introduce both energy consumption overhead—especially detrimental to satellite downlinks and IoT device communications—and operational delays that compound propagation latency in long-distance networks.

To enhance the reliability and reduce latency of communication systems, Hybrid Automatic Repeat Request (HARQ) schemes \cite{chase1985code,kallel1990analysis} employing various strategies are widely used.

\subsection{Background and motivation}
HARQ schemes \cite{Costello2004,zurawski2014industrial} combine forward error correction (FEC) coding with cyclic redundancy check (CRC) to ensure reliable data transmission. In HARQ schemes, if the receiver detects an uncorrectable error in a codeword, it feeds back a negative acknowledgment signal to request retransmission. Otherwise, the receiver sends an acknowledgement signal. Two primary classes of HARQ schemes are widely adopted in practice. The first one is the HARQ Chase Combining (HARQ-CC) scheme \cite{chase1985code}, in which each retransmission codeword is identical to the original one. The other is the HARQ Incremental Redundancy (HARQ-IR) scheme \cite{kallel1990analysis}, in which the codewords transmitted over the first transmission and subsequent transmissions are generally different. If all the message and parity bits are included in every retransmission, then such an HARQ-IR scheme is called self-decodable. The receiver combines multiple received codewords into a long codeword to enhance the error-correction capability. Here, we consider the HARQ-IR scheme, which is self-decodable and has the same codeword size for each retransmission. That is, a $[mn, k]$ linear code is partitioned into $m$ segments, each transmitted sequentially. Each segment can be decoded independently or combined into a longer code to improve error correction. In this work, we refer to such linear codes as \textit{IR-linear codes}. From an error-correction perspective, the HARQ-CC and self-decodable HARQ-IR schemes can be mathematically modeled as two repeated communication frameworks, as shown in Figure \ref{fig_retransmission channels}.

\begin{figure}[htbp]
	\centering
	\begin{tabular}{c}
		\includegraphics[width=80mm]{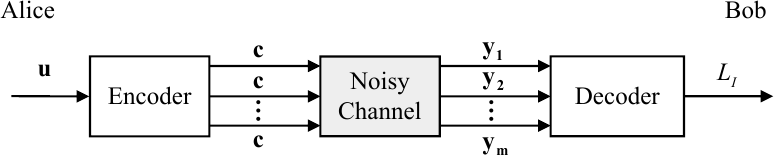} \\
		{\footnotesize (a) Type-I repeated communication model} \\
		\includegraphics[width=80mm]{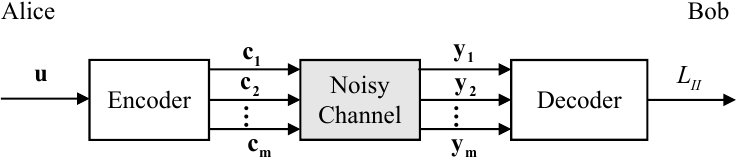} \\
		{\footnotesize (b) Type-II repeated communication model}
	\end{tabular}
	\caption{ 
		Two types of repeated communication channels:  
Alice repeatedly transmits the same encoded version $\mathbf{c}$ (resp. different encoded version $\mathbf{c}_i$) of message $\mathbf{u}$ to Bob through a noisy channel. Bob decodes the received codewords $\mathbf{y}_i$ to generate a list of candidate codewords $L_I$ (resp.\ $L_{II}$). If $L_I$ or $L_{II}$ contains the correct message $\mathbf{u}$, then decoding succeeds.
}
	\label{fig_retransmission channels}
\end{figure}

 FEC has a significant impact on the performance of HARQ schemes. 
A well-designed FEC can enhance the error-correcting capability of HARQ, reduce the number of retransmissions, and thereby reduce latency.
In \cite{Wang2003OnRT}, Wang and Orchard lowered the retransmission probability by employing embedded coset codes.
Li et al. \cite{li2010retransmission} introduced a method that modulates the transmitted codeword onto different subcarriers with different permutations, which demonstrated strong performance in overlapping channels.
Researchers have also proposed HARQ schemes based on turbo codes \cite{Zhang2020}, LDPC codes \cite{Kim2020}, and polar codes \cite{Mohammadi2017}. 
Given the critical application of HARQ in communication systems, this topic has been extensively studied, as summarized in the review literature \cite{Ahmed2021}.

In advanced storage systems, high-density data applications (e.g., molecular data storage and chemical informatics) encounter challenges analogous to iterative transmission models. 
This corresponds to transmitting identical codewords across $m$ independent channels. For a linear code with minimum Hamming distance $d$, Levenshtein~\cite{levenshtein2001efficient} proved that if $t > \lceil (d-1)/2 \rceil$, successful decoding requires at least $m$ channels satisfying
\begin{align}\label{Eq_Levenshtein}
    m > \sum_{i=0}^{t - \lceil (d-1)/2 \rceil -1} \binom{n - d}{i} \sum_{k=i + d - t}^{t - i} \binom{d}{k}.
\end{align}

However, as $t$ increases beyond $\lceil (d-1)/2 \rceil +1$, $m$ grows rapidly to the level of message length $n$. 
For example, binary Golay code has parameters $[23,12,7]_2$, which can correct $3$ bit errors. 
By Equation (\ref{Eq_Levenshtein}), we have $m> 1+\binom{7}{3}=36$ for $t=4$. 

To address these challenges, the symbol-pair read channel and symbol-pair codes~\cite{Cassuto2011CodesFS} were proposed for high-density storage.  
In the symbol-pair read channel, the probability of errors occurring in adjacent bits are correlated. 
Therefore, these codes generalize classical linear codes by defining the symbol-pair weight of $\mathbf{c}$ as the number of consecutive non-zero symbol pairs. 
Yaakobi et al.~\cite{Yaakobi2016ConstructionsAD} further extended this to $b$-symbol read channels, proving that a binary cyclic code with minimum Hamming distance $d$ has minimum symbol-pair distance at least $d+\lceil d/2 \rceil$. 
In \cite{shi2023connections}, Shi et al. proved that $[n,k,d]_q$ cyclic code  has a minimum 
$b$-symbol distance $\geq \sum_{i=0}^{b-1} \lceil d/q^i \rceil$. 
Since $b$-symbol codes can correct any error pattern correctable by classical repetition codes, $b$-symbol codes can be regarded as an enhanced version of classical repetition codes. 
They can be used to directly enhance the performance of Type-I repeated communication systems. 
 Nevertheless, three limitations exist:
\begin{itemize}
    \item The minimum $b$-symbol distance grows slowly and almost linearly as $b$ increases. 
    \item If a specific bit repeatedly fails during reading or transmission, the weight of the corresponding $b$-symbol error will be $b$. 
    \item For general linear codes, one can only be proved that $d_b\ge d_{b-1}+1$.
\end{itemize} 

Thus, $b$-symbol codes offer diminishing coding gains as $b$ increases. 
To verify this, we simulate the frame error rate (FER) of the binary cyclic Golay $[23,12,7]_2$ code under the $b$-symbol metric in an additive white Gaussian noise (AWGN) channel\footnote{Note that, in AWGN, the bit error probabilities of different bits are independent. Therefore, the simulation results cannot illustrate the performance of $b$-symbol codes in the $b$-symbol read channel.} using the minimum $b$-symbol distance decoding algorithm. The relationship between FER and different values of $b$ is illustrated in Figure  \ref{FER_of_b_symbol_Golay_codes}. 
The simulation results show that for $b\le 5$, the coding gain of the $b$-symbol metric decreases gradually. When $b>5$, the $b$-symbol metric provides only negative coding gain.

\begin{figure}[h]
	\centering
	\includegraphics[scale=0.3]{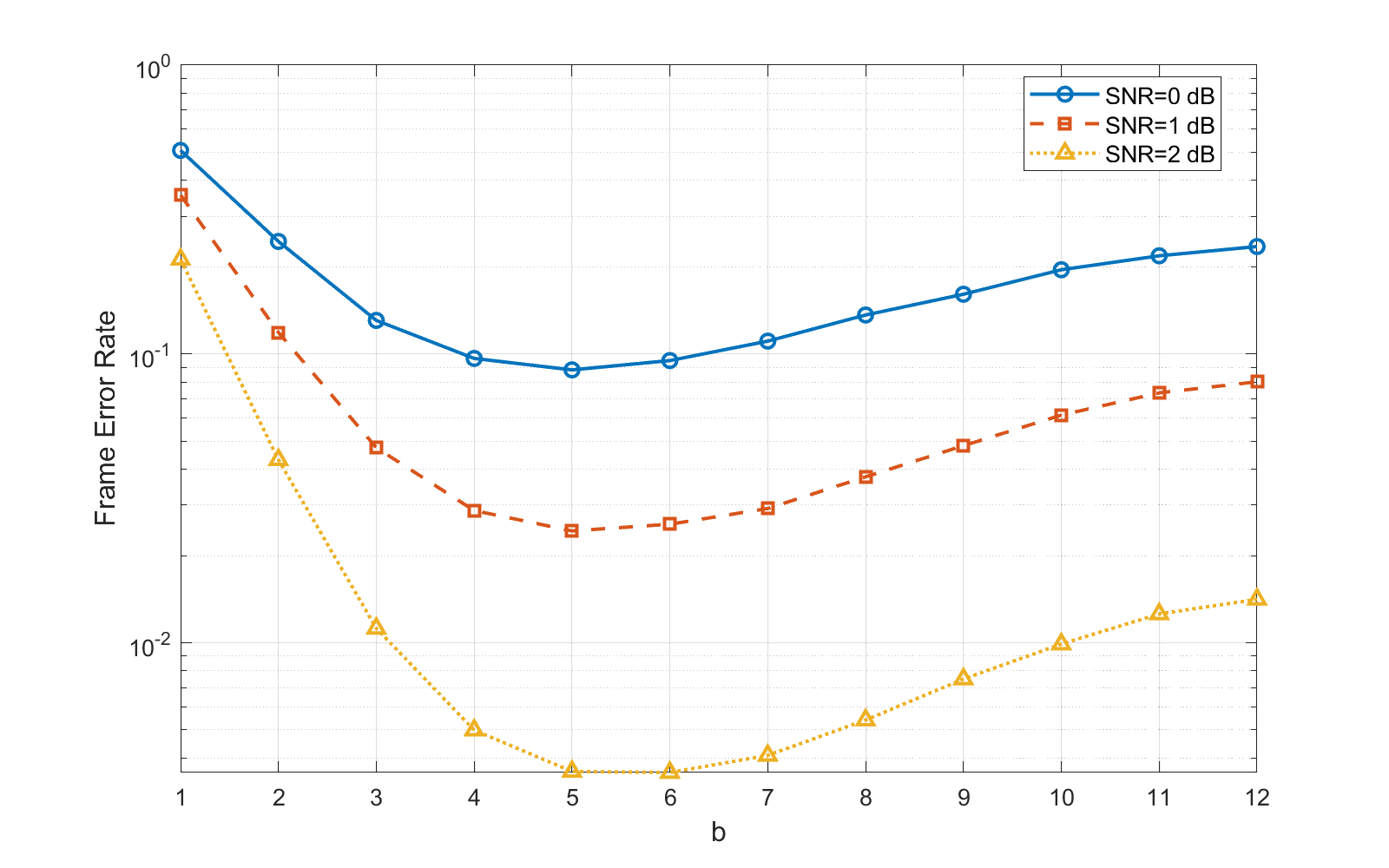}
	\caption{FER curves of binary cyclic Golay codes under $b$-symbol metric with different $b$: The signal-to-noise ratio (SNR) quantifies the ratio of signal power to noise power, expressed in decibels (dB) as \( 10 \log_{10}(P_{\text{signal}} / P_{\text{noise}}) \).}
\label{FER_of_b_symbol_Golay_codes}
\end{figure}

	From a coding theory perspective, the aforementioned methods enhance the fault tolerance of the two repeated communication models by treating multiple retransmitted codewords either as a single codeword of a linear repetition code, as a longer codeword in a code with the Hamming metric, or as a codeword measured by the $b$-symbol metric.
	Therefore, a natural question emerges: \textbf{whether linear repetition codes can have a greater error-correction capability}. 
	For Type-I repeated transmission scenario, is there a better encoding method, and if so, can similar approaches enhance the coding gain of its Type-II counterpart? 
	This motivates us to generalize the concept of linear repetition codes and investigate their constructions, their bounds and decoding methods.

\subsection{Contributions}
The following is a brief overview of the contributions made by this paper.

\begin{enumerate}
	\item 
	We introduce the concept of generalized repetition codes (GRCs) and classify them into Type-I and Type-II, corresponding to the two repeated communication models. 
	We regard GRCs as error-correcting codes with multiple metrics and multiple minimum distances.
	This design enables GRCs to perform multi-round error correction under different metrics, achieving stronger error-correction capabilities than classical error-correcting codes.
	Simulation results show that under multi-round decoding, GRCs outperform classical repetition codes, $b$-symbol codes, and IR-linear codes in performance. 
	\item 
	By leveraging the special algebraic structure of Type-I GRCs, we derive a Griesmer-type bound for them and provide explicit construction methods for optimal Type-I GRCs using cyclic codes, extended cyclic codes, and quasi-cyclic codes. 
	\item 
	We establish a mathematical relationship between general linear codes and Type-II GRCs, enabling the derivation of optimal Type-II GRCs.
	Finally, we construct numerous Type-II quasi-cyclic GRCs with good parameters via a systematic algorithmic search.
\end{enumerate}
\subsection{Organization}

The rest of this paper is organized as follows. Section \ref{Sec II} introduces the necessary preliminaries. Section \ref{Sect III} presents the definitions and metrics for Type-I and Type-II GRCs. Sections \ref{Sect IV} and \ref{Sect V} address the bounds and constructions of the Type-I and Type-II GRCs, respectively. Finally, Section \ref{Sect VI} concludes this paper. 

	\section{Preliminaries}\label{Sec II}
Let $m$ be a positive integer, and let $[m]$ denote the set ${1,2,\ldots,m}$. 
Let $q$ be a prime power and let $\mathbb{F} _q$ denote the finite field with $q$ elements. 
Let $\mathbf{v}=(v_1,v_2,\ldots,v_{mn})$ be a vector in $\mathbb{F} _q^{mn}$, where $n$ is a positive integer.
The block\footnote{In some literature, linear codes are also referred to as linear block codes, where ``block" denotes the code length \( n \). To avoid confusion, we clarify here that in this paper, "block" refers to the partitioning of the code length into segments of size \( m \), not \( n \).
} weight of $\mathbf{v}$ is 
\begin{equation}
	\mathrm{w}_m(\mathbf{v})=\#\left\{i \mid (v_{i},v_{n+i},\ldots,v_{(m-1)n+i}) \neq\mathbf{0}, 1 \leq i \leq n\right\}.
\end{equation}
When $m=1$, the block weight is indeed the same as the Hamming weight.
The block distance between $\mathbf{v_1},\mathbf{v_2}\in \mathbb{F} _q^{mn}$ is $\mathrm{d}_m(\mathbf{v_1},\mathbf{v_2})=\mathrm{w}_m(\mathbf{v_1}-\mathbf{v_2})$.
A linear subspace $\mathcal{C}$ of $\mathbb{F}_q^{mn}$ with minimum block distance $d_m$ is called an $[(n,m),k,d_m]_q$ linear block metric code, where $m$ is the block size.
Such a code can correct up to $\lfloor \frac{d_m-1}{2} \rfloor$ or detect $d_m-1$ block errors.
The vectors in $\mathcal{C}$ are called codewords. 
A $k \times mn$ matrix $G$ consisting of $k$ basis vectors of $\mathcal{C}$ is a generator matrix of $\mathcal{C}$.
If $m=1$, the block metric is reduced to the Hamming metric. In this case, we denote $\mathcal{C}$, $\mathrm{w}_m(\cdot)$,  $\mathrm{d}_m(\cdot)$, and parameters $[(n,m),k,d_m]_q$ by $C$, $\mathrm{w}(\cdot)$,  $\mathrm{d}(\cdot)$, and $[n,k,d]_q$, respectively.
   
    For $t$ linear codes $C_1$, $C_2$, $\ldots$, and $C_t$ of the same dimension, we denote by $C_1\times C_2\times\cdots\times C_t$ their juxtaposition code, i.e., 
 \begin{equation}
 C_1\times C_2\times\cdots\times C_t =\left\{\left(\mathbf{c}_1,\mathbf{c}_2,\ldots,\mathbf{c}_{t}\right)\,\bigg|\,\mathbf{c}_i\in C_i, i\in [t] \right \}.
 \end{equation}

The extended code of an $[n,k]_q$ linear code $C$ is
\begin{equation}
	\widehat{C} =\left\{\left(c_1,c_2,\ldots,c_{n},-\sum_{i=1}^{n}c_i \right)\,\bigg|\,\left(c_1,c_2,\ldots,c_{n} \right)\in C \right \}.
\end{equation} 
The support of a codeword $\mathbf{c}=\left(c_1,c_2,\ldots,c_{n} \right)$ of $C$ is $\text{Supp}(\mathbf{c})=\{i\mid c_i\ne 0, 1\le i \le n \}$.
The support of $C$ is $\text{Supp}(C)=\bigcup _{c \in C} \text{Supp}(\mathbf{c})$. Let $n(C)=|\text{Supp}(C)|$ be the \textit{effective length} of $C$.
If $n=n(C)$, then we call $C$ is \textit{full length}.

For $\mathbf{v}=(\mathbf{v}_1,\mathbf{v}_2,\ldots,\mathbf{v}_m)\in \mathbb{F}_q^{mn}$ and $\mathbf{v}_i\in \mathbb{F}_q^{n}$, we define 
	$
	\phi_m(\mathbf{v})$ as an $m\times n$ matrix $\begin{pmatrix}
		\mathbf{v}_1^T, \mathbf{v}_2^T, \ldots, \mathbf{v}_m^T 
	\end{pmatrix}^T
$. 
Then, $\mathrm{w}_m(\mathbf{v})$ is equal to the number of nonzero columns of $\phi_m(\mathbf{v})$.  
 
As shown in \cite{Cassuto2011CodesFS,Yaakobi2016ConstructionsAD}, for a vector $\mathbf{v}=(v_1,v_2,\ldots,v_{n})\in\mathbb{F}_q^n$ and an integer $b$ with $1 \le b \le n$, the $b$-symbol metric defines a $b$-symbol vector $$\mathbf{v}_b=(v_1,v_2,\ldots,v_{b},v_2,v_3,\ldots,v_{b+1},\ldots, v_{n},v_1,\ldots,v_{b-1})\in \mathbb{F}_q^{bn}.$$

The $b$-symbol weight of $\mathbf{v}_b$ is defined as $\mathrm{w}_b(\mathbf{v}_b)=\#\left\{ i\mid (v_{i},\ldots,v_{i+b-1})\ne \mathbf{0}, 1\le i\le n  \right\}$, where the indices are taken modulo $n$. 
Via an appropriate coordinate permutation, any $[n,k]_q$ linear code $C$ with minimum $b$-symbol distance $d_b$ is equivalent to an $[(n,b),k,d_b]_q$ linear block metric code whose generator matrix is the block matrix $(G,GX,\ldots,GX^{b-1})$, where $X$ is the cyclic shift matrix of order $n$.
 
\subsection{Bounds on linear block metric codes}

\begin{lemma}(Singleton Bound, \cite{huffman2010fundamentals})\label{Singleton_bound}
	If $\mathcal{C}$ is an $[(n,m),k,d_m]_q$ linear block metric code, then 
	\begin{equation}
		d_m\le n- \frac{k}{m} +1.
	\end{equation} 
\end{lemma}

The code attaining the bound in Lemma \ref{Singleton_bound} is called \textit{maximum distance separable (MDS)} code.   
If $d_m= n-\left\lceil \frac{k}{m}\right\rceil+1$, then we call $\mathcal{C}$ a \textit{fractional MDS} code \cite{ball2025griesmer}. 

\begin{lemma}(Singleton Bound for $b$-symbol code, \cite{Ding2018MaximumDS})\label{Singleton_bound_b}
	If $\mathcal{C}$ is an $[(n,b),k,d_b]_q$ $b$-symbol code, then 
	\begin{equation}
		d_b\le n- k +b.
	\end{equation} 
\end{lemma}

The Griesmer bound is essential in determining the optimality of linear codes.  The following is a generalization of Griesmer bound for linear block metric codes. 

\begin{lemma}\label{Griesmer_Bound}(Griesmer Bound,\cite{Guan2023SomeGQ,Luo2024GriesmerBA}) 
	If $\mathcal{C}$ is an $[(n,m),k,d_m]_{q}$ linear block metric code, for $k\ge m\ge1$, the following holds: 
	\begin{equation}\label{AGB}
		nN_m \geq g_{q,m}(k, d_m)=\sum\limits_{i=0}^{k-1}\left\lceil  \frac{q^{m-1}d_m}{q^{i}}\right\rceil, 
	\end{equation} 
where $N_m\triangleq\frac{q^m-1}{q-1}$.
\end{lemma}

The code satisfying this bound with equality is called \textit{Griesmer code}. In this paper, we mainly consider the case $m\ge 2$, and when $m = 1$, we denote $g_{q,m}(k, d_m)$ by $g_q(k, d_m)$.

If there is no $[(n,m),k,d_m+1]_q$ code, then an $[(n,m),k,d_m]_q$ code is called \textit{optimal}. 
Since an $[n,k_l,d_l]_{q^m}$ linear code over $\mathbb{F}_{q^m}$ is also an $[(n,m),mk_l,d_l]_q$ linear block metric code, an $[(n,m),k,d_m]_q$ linear code \textit{outperforms  $\mathbb{F}_{q^m}$-linear codes} if there is no $[(n,m),m\cdot\lceil \frac{k}{m} \rceil,d_m]_q$ code. 

\subsection{The geometry of linear block metric codes}
A $k$-dimensional vector space over $\mathbb{F} _q$ is viewed as the projective geometry of dimension $k-1$, denoted by $PG(k-1,q)$. 
An $(m-1)$-flat is a projective subspace of
dimension $m-1$ in $PG(k-1, q)$. $0$-flats, $1$-flats, and $(k - 2)$-flats are called points, lines, and hyperplanes, respectively.

Let $\mathcal{F}_m$ be an $(m-1)$-flat and $S_{\mathcal{F}_m}$ be the set of all $(m-1)$-flats of $PG(k-1,q)$.
Then, we have  $\theta_{k,m}=|S_{\mathcal{F}_m}|=\frac{\left(q^{k}-1\right)\left(q^{k-1}-1\right) \ldots\left(q^{k-m+1}-1\right)}{\left(q^{m}-1\right)\left(q^{m-1}-1\right) \ldots(q-1)}$ and $N_m=|\mathcal{F}_m|=\frac{q^m-1}{q-1}$.
The linear block metric Simplex code \cite{bierbrauer2017introduction} with parameters $[(\theta_{k,m},m),k,q^{k-m}\theta_{k-1,m-1}]_{q}$ is described
by $S_{\mathcal{F}_m}$.
All points of $PG(k-1,q)$ constitute the well-known Simplex code with parameters
$[\frac{q^k-1}{q-1},k,q^{k-1}]_q$.  
According to \cite{huffman2021concise}, both Simplex and Hamming codes are equivalent to cyclic codes. 
Moreover, all nonzero vectors in $\mathbb{F}_q^k$ also form a $[q^k-1,k,q^k-q^{k-1}]_q$ cyclic code\footnote{The BCH code with defining set $T_n\setminus C_1$, for details see Chapter 5 of \cite{huffman2010fundamentals}.}. 
An $[(n,m), k, d]_{q}$ linear block metric code $\mathcal{C}$ is equivalent to a collection $\Gamma $ of $n$ $(m-1)$-flats of $PG(k-1,q)$, such that any hyperplane $H$ contains at most $n-d$ of those $(m-1)$-flats.

\subsection{Permutations and permutation groups} 
A permutation is a bijective function \( \sigma: S \to S \), where \( S = [n] \).
All permutations form the \textit{symmetric group} \( Sym_n \), which has order \( |Sym_n| = n! \). 
A $k$-cycle $(a_1,a_2,\ldots,a_k)$ represents a permutation where: $a_1\to a_2$, $a_2\to a_3$, $\ldots$, $a_k\to a_1$ and other elements are fixed. 
Every permutation \( \sigma \in Sym_n \) admits a unique decomposition into disjoint cycles. 
Its \textit{cycle type} is specified by the sequence $(i^{\theta_i})_{i\in\mathbb{N}}$ where $\theta_i$ counts the number of cycles of length $i$ in this decomposition. 
We refer to the length of the longest cycle in the decomposition of \( \sigma \) as its \textit{maximum cyclic length}.
 For example, permutation $\sigma = \begin{pmatrix}
	1&2  & 3 & 4 &5 \\
	3& 2 & 1 &  5&4
\end{pmatrix}$ can be denoted by $\sigma=(1,3)(2)(4,5)$. 
Its cycle type and maximum cyclic length are $(1^1,2^2)$ and $2$, respectively. 
Two permutations \( \sigma, \tau \in Sym_n \) are \textit{equivalent} if they share the same cycle type, i.e., there exists \( \rho \in Sym_n \) such that $\tau = \rho \sigma \rho^{-1}$. 
The order of permutation $\sigma$ is the smallest integer $r$ satisfying  that $\sigma^r$ is the identity element of $Sym_n$.  
  For $\mathbf{v}=(v_1,v_2,\ldots,v_{n})\in \mathbb{F}_q^n$, we define  $\sigma(\mathbf{v})=(v_{\sigma(1)},v_{\sigma(2)},\ldots,v_{\sigma(n)})$.

\subsection{Cyclic and quasi-cyclic codes}

Let $\pi$ be the cyclic shift permutation in $Sym_n$, where \(\pi(i) = (i+1) \bmod n\) for all \(i \in [n] \).  
A linear code \(C\) of length \(n\) over \(\mathbb{F}_q\) is called a \textit{cyclic code} if it is invariant under $\pi$, i.e., \(\pi(\mathbf{c}) \in C\) for every codeword \(\mathbf{c} \in C\).  
This property allows cyclic codes to be algebraically represented as ideals in the polynomial ring $\mathbb{R}_{q,n} =\mathbb{F} _{q}[x] /\left\langle x^{n}-1\right\rangle$. 
Each codeword corresponds to a polynomial $c(x) = c_0 + c_1 x + \dots + c_{n-1} x^{n-1}$, and cyclic shifts correspond to multiplication by $x$ modulo $x^n - 1$. 
A cyclic code is uniquely defined by its generator polynomial $g(x)$, which satisfies $g(x) \mid (x^n - 1)$ and $\deg(g(x))=n - k$, where $k$ is the dimension of the code. The parity-check polynomial is given by $h(x) = (x^n - 1)/g(x)$. 
Suppose \(h(x)=h_1(x)\cdots h_t(x)\), where \(h_i(x)\) is irreducible over $\mathbb{F}_q$. We define 
\begin{equation}
	\kappa(g(x))=\min\{\deg(h_i(x)) \mid \deg(h_i(x))\ge2, i \in [t] \}.
\end{equation}
It follows that $C$ has no cyclic subcodes whose dimension lies strictly between $2$ and $\kappa(g(x))$.

Let $\pi_\ell$ be the permutation consisting of $\ell$ disjoint $n$-cycles:   $\pi_\ell=(1,\ldots,n)(n+1,\ldots,2n)\cdots((\ell-1)n+1,\ldots,\ell n )$. 
A linear code $C$ of length $\ell n$ over $\mathbb{F}_q$ is a quasi-cyclic code of index $\ell$ if $C$ is invariant under permutation $\pi_\ell$, and $\ell$ is the smallest integer with this property.
The generator matrix of quasi-cyclic code $C$ consists of connected circulant matrices, and each circulant matrix corresponds to a generator polynomial. 
According to \cite{LALLY2001157}, the generator polynomial
matrix of quasi-cyclic code $C$ can be written as an upper triangular polynomial matrix  
$$\mathbf{g}(x) =\begin{pmatrix}
	g_{1,1}(x)&\cdots   &g_{1,\ell}(x) \\
	& \ddots  & \vdots \\
	&   &g_{h,\ell}(x)
\end{pmatrix},$$ where $g_{i,j}\in \mathbb{R}_{q,n}$ and the empty entries denote zero. For precision, we refer to quasi-cyclic code $C$ generated by $\mathbf{g}(x)$ as $(h,\ell)$-quasi-cyclic code.

\section{Generalized repetition codes and their application to two repeated communication models}
\label{Sect III} 
In this section, we introduce the basic definitions of Type-I and Type-II GRCs. 
We then propose combinatorial block metrics and multi-round decoding methods for GRCs. 
Additionally, we show their applications in the two types of repeated communication models. Finally, we conduct simulations to demonstrate the advantages of GRCs over classical repetition codes, $b$-symbol codes, and IR-linear codes.

	\subsection{Generalized repetition codes}
Let $Ps(n,\mathbb{F}_q)$ denote the set of all permutation matrices of size $n\times n$ over $\mathbb{F}_q$. 
Let $GL(n,\mathbb{F}_q)$ denote the general linear group of size $n$ over $\mathbb{F}_q$. 
We now give a generalized definition of linear repetition codes.

	\begin{definition}\label{D_repetition code}
	Let $C$ be an $[n,k]_q$ linear code with generator matrix $G$. 
	If $\mathcal{C}$ is an $[(n,m),k]_q$ linear code with generator matrix 
	\begin{equation}\label{eq_Type-GRC}
		\mathcal{G}_m=(G,GA_1, \ldots,GA_{m-1}),
	\end{equation}
	or
	\begin{equation}
		\mathcal{G}_m^{\prime}=(G,B_1G, \ldots,B_{m-1}G),
	\end{equation}
	where $A_i \in Ps(n,\mathbb{F}_q)$ and $B_i \in GL_k(\mathbb{F}_q)$,
	then $\mathcal{C}$ is called a \textit{generalized repetition code (GRC)} of $C$. If all $A_i$ and $B_i$ are identity matrices, then $\mathcal{C}$ reduces to a classical linear repetition code. Specifically, $\mathcal{C}$ is called a \textit{Type-I GRC}\footnote{Note that all Type-I GRCs are equivalent (regardless of the permutations $A_1,\ldots,A_m$) over discrete memoryless channels. For details on code equivalence, see \cite{huffman2010fundamentals}.} if it is generated by $\mathcal{G}_m$, and a \textit{Type-II GRC} if generated by $\mathcal{G}_m^{\prime}$.
	 
 
 	For $i \ge 2$, if $A_i = A_1^{i-1}$ or $B_i = B_1^{i-1}$, then $\mathcal{C}$ is called \textit{regular}. 
	Moreover, if the matrices \( I_n, A_1, \ldots, A_{m-1} \) or \( I_k, B_1, \ldots, B_{m-1} \) are linearly independent over $\mathbb{F}_q$, then \( \mathcal{C} \) is called \textit{non-degenerate}; otherwise, it is called \textit{degenerate}.
	All GRCs studied in this paper are non-degenerate.
\end{definition}


By applying suitable coordinate transformations, 
 $b$-symbol codes can be converted to Type-I regular GRCs with $A_i=X^i$, where $X$ is a cyclic shift matrix of order $n$. 
Moreover, Type-I GRCs employ distinct metrics, which can increase the coding gain with the number of retransmissions.
Since each $A_i$ is a permutation matrix, Type-I GRCs have a minimum Hamming distance of $m \cdot \mathrm{d}(C)$.

We now define the combinatorial block metrics for Type-I and Type-II GRCs.

\begin{definition}\label{Def_sub-block distance}
	Let $1\le r\le m$ be an integer and let $\binom{[m]}{r}$ denote the set consisting of all $r$-element subsets of $[m]$.
	For an $[(n,m),k]_q$ GRC $\mathcal{C}$ and $T=\{i_1,i_2,\ldots,i_r\} \in \binom{[m]}{r}$, we define the \textit{ $r$-th sub-block code} $\mathcal{C}^T$ of $\mathcal{C}$ as
	\begin{equation}
		\mathcal{C}^T=\{\mathbf{c}^T=(\mathbf{c}_{i_1},\mathbf{c}_{i_2},\ldots,\mathbf{c}_{i_r}) \mid \mathbf{c}=(\mathbf{c}_1,\mathbf{c}_2,\ldots,\mathbf{c}_m)\in \mathcal{C}, \text{ where } \mathbf{c}_i\in \mathbb{F}_q^n\}.
	\end{equation}

	
	We define the \textit{$r$-th minimum sub-block distance} of $\mathcal{C}$   as
	\begin{equation}
		\mathbf{d}_r(\mathcal{C})=\min \left\{ \mathrm{d}_r(\mathcal{C}^T) \,\bigg|\, T\in \binom{[m]}{r}  \right\}. 
	\end{equation}

	 The \textit{$r$-th minimum sub-Hamming distance} of $\mathcal{C}$ is defined as 
	\begin{equation}
		\underline{\mathbf{d}}_r(\mathcal{C})=\min \left\{ \mathrm{d}(\mathcal{C}^T) \,\bigg|\, T\in \binom{[m]}{r}  \right\}. 
	\end{equation}

	We call the sequences $(\mathbf{d}_1(\mathcal{C}),\mathbf{d}_2(\mathcal{C}),\ldots,\mathbf{d}_m(\mathcal{C}))$ and $(\underline{\mathbf{d}}_1(\mathcal{C}),\underline{\mathbf{d}}_2(\mathcal{C}),\ldots,\underline{\mathbf{d}}_m(\mathcal{C}))$ \textit{sub-block distance hierarchy (SBDH)} and \textit{sub-Hamming distance hierarchy (SHDH)} of $\mathcal{C}$, respectively. 
	If for all $r$, there is no $[(n,r),k,\mathbf{d}_r(\mathcal{C})+1]_q$ Type-I (resp. Type-II) GRC, then we call $\mathcal{C}$ a \textit{SBDH-optimal} Type-I (resp. Type-II) GRC. 
	Since permutations do not increase the Hamming distance of Type-I GRCs, we only consider the SHDH of Type-II GRCs. 
	If for all $r \in [m]$, there is no $[(n,m),k]_q$ Type-II GRC with $r$-th minimum sub-Hamming distance $\underline{\mathbf{d}}_r(\mathcal{C})+1$, then we call $\mathcal{C}$ a \textit{SHDH-optimal} Type-II GRC. 
\end{definition}

Definition \ref{Def_sub-block distance} formalizes the combinatorial metrics for GRCs.  
The decoding of a GRC $\mathcal{C}$ can be decomposed into decoding its $r$-th sub-block codes $\mathcal{C}^T$, for $r \in [m]$. 
Since $\mathcal{C}$ has $2^m - 1$ distinct subsets, the decoding process involves up to $2^m - 1$ sub-block codes. 
Unlike classical repetition codes under the Hamming metric, it is advantageous to view a codeword $\mathbf{c}$ of a GRC $\mathcal{C}$ as an array, i.e., \(\phi_m(\mathbf{c})\). 
The error correction mechanism addresses up to $\lfloor \frac{\mathbf{d}_r(\mathcal{C})-1}{2} \rfloor$ column errors in any \(r\) rows of \(\phi_m(\mathbf{c})\), resembling array codes. 
If GRC $\mathcal{C}$ is of Type-II, then it is also able to correct up to $\lfloor \frac{\underline{\mathbf{d}}_r(\mathcal{C})-1}{2} \rfloor$ bit errors in any \(r\) rows of \(\phi_m(\mathbf{c})\). This combinatorial metric exponentially amplifies the number of correctable error patterns of GRCs.  
Table \ref{Comparison of Error Correction Capabilities} compares the number of error patterns that can be corrected for the Type-I and Type-II GRCs, classical repetition codes, $b$-symbol codes, and IR-linear codes. 
Therefore, in CRC-aided communication environments, GRCs exhibit stronger error-correction capability.

\begin{table*}[ht]
	\caption{Comparison of Correctable Error Patterns of  Different Types of Codes}\label{Comparison of Error Correction Capabilities}
	\centering
	\resizebox{\textwidth}{!}{
			\begin{tabular}{ccccc}
				\toprule
				Codes Type        & \makecell[c]{Correctable Error Patterns \\ Under Hamming Metric} & \makecell[c]{Correctable Error Patterns \\Under Block Metric ($2\le r\le m$)} & Chase Combining & \makecell[c]{Total of Correctable\\ Error Patterns} \\ \midrule
				Type-I GRC     &                           $m$                            &                            $2^m-m-1$                            &  $\checkmark$   &                        $2^m$                        \\
				$b$-symbol codes &                         $m$                          &     $1$                                                  &    $\checkmark$     &      $m+2$                                 \\
				Classic repetition Codes &                            $m$                             &                            $\times$                             &  $\checkmark$   &                        $m+1$                        \\
				Type-II GRC     &                         $2^m-1$                          &                             $2^m-m-1$                             &    $\times$     &                           $2^{m+1}-m-2$                    \\
				IR-linear codes &                         $2m-1$                          &     $\times$                                                  &    $\times$     &      $2m-1$                                 \\
				\bottomrule
			\end{tabular}%
		}
	\end{table*}

	The retransmission framework using the Type-I GRC is shown in Figure  \ref{packet_tran II}, where the ``Permutator" is determined by $A_i$ and ``Combiner" is implements Chase Combining.	
	Similarly, the retransmission scheme using the Type-II GRC is shown in Figure  \ref{packet_tran III}, where the ``Linear Transformator" is defined by $B_i$.	  
	If Bob decodes the received codewords solely based on the Hamming metric, these schemes reduce to the repeated communication models shown in Figure \ref{fig_retransmission channels}.

	\begin{figure}[h]
		\centering
		\includegraphics[width=115mm]{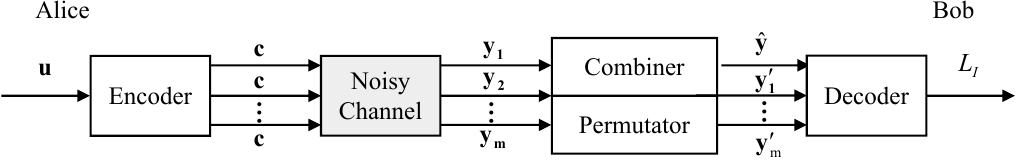}
		\caption{Modified Type-I repeated communication model using Type-I GRCs: 
			Once Bob receives $m\ge 2$ codewords: $\mathbf{y}_1,\ldots,\mathbf{y}_{m}$, Bob takes Chase Combining and permutations to them get $\hat{\mathbf{y}}$ and $\mathbf{y}_i^{\prime}=\mathbf{y}_iA_{i-1}$, respectively, where $i\ge 2$ and $A_i\in Ps_{n}(\mathbb{F}_q)$.
		}
		\label{packet_tran II}
	\end{figure}

	\begin{figure}[h]
		\centering
		\includegraphics[width=120mm]{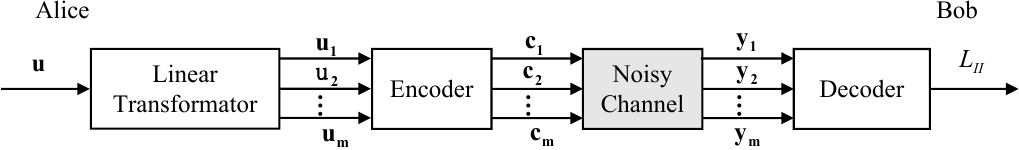}
		\caption{Modified Type-II repeated communication model using Type-II GRCs:  Alice performs a non-singular linear transformation $B_i$ on message $\mathbf{u}$ before retransmitting it each time. i.e.,  $\mathbf{c}_i=\mathbf{u}B_{i-1}G$ for $i\ge 2$, where  $B_i \in GL_k(\mathbb{F}_q)$. 
		}
		\label{packet_tran III}
	\end{figure}

\subsection{Basic idea for the decoding of generalized repetition codes}

In this subsection, we introduce a basic decoding method for GRCs. 
Without loss of generality, assume $\mathbf{z}=(\mathbf{z}_0,\mathbf{z}_1,\ldots,\mathbf{z}_{m-1})$ is the sequence received by Bob after $m$ transmissions in the above two schemes. 
According to Definition \ref{Def_sub-block distance}, the decoding of GRC  $\mathcal{C}$ under combinatorial metrics is equivalent to decoding $[(n,|T|),k]_q$ sub-block codes $\mathcal{C}^T$ of $\mathcal{C}$ for all non-empty subsets $T\in[m]$. 
If $\mathbf{z}$ is a codeword of Type-I GRC, then we only need to decode $\mathbf{z}^{T}$ once under block metric;
if $\mathbf{z}$ is a codeword of Type-II GRC, then we need to decode $\mathbf{z}^{T}$ under both Hamming and block metrics.

In practical communication, codewords typically include CRC bits, enabling users to verify the correctness of the decoded result after decoding. Therefore, the combinatorial metric structure of GRCs enables multi-round decoding, where decoding stops early once a decoded result passes CRC. 
As decoding complexity increases with $|T|$, the process proceeds iteratively by increasing $|T|$ until reaching \(m \). 
For convenience of exposition, we define the maximum value of \( |T| \) as the \textit{decoding depth}. In the following example, we simulate the performance of GRCs constructed from Gray codes at different decoding depths.

\begin{example}
	Set $q=2$ and $n=23$. 
	The binary Golay code $\mathcal{G}_{23}$ is a $[23,12,7]_2$ cyclic code with generator polynomial  $g(x)=x^{11} + x^9 + x^7 + x^6 + x^5 +  x + 1$. 
	Let $f_1(x)=x^9 + x^6 + x^5 + x^4 + x^3 + x + 1$, $f_2(x)=x^{12} + x^{11} + x^{10} + x^9 + x^8 + x^5 + x$, and $f_3(x)=x^{11} + x^8 + x^7 + x^6 + x^5 + x^3 + x$.  
	Let $\mathbf{g}_1(x)=(g(x),xg(x),x^2g(x),x^3g(x))$ and $\mathbf{g}_2(x)=(g(x),f_1(x)g(x),f_2(x)g(x),f_3(x)g(x))$. 
	
	
	Let $\mathcal{C}_I$ and $\mathcal{C}_{II}$ be two quasi-cyclic codes generated by $\mathbf{g}_1(x)$ and $\mathbf{g}_2(x)$, respectively. 
	Since $\gcd(f_i(x),x^n-1)=1$, $f_i(x)g(x)$ also generates $\mathcal{G}_{23}$ and $\mathbf{d}_1(\mathcal{C}_I)=\mathbf{d}_1(\mathcal{C}_{II})=7$. 
	Moreover, it is easy to use Magma to compute that $\mathcal{C}_I$ has SBDH $(7,11,13,15)$, and the SBDH and SHDH of $\mathcal{C}_{II}$ are $(7,12,16,19)$ and $(7,14,24,36)$, respectively.

	Suppose Alice and Bob use $\mathcal{C}_I$ and $\mathcal{C}_{II}$ to transmit messages. 
	Thanks to the special structure of GRCs, Bob can perform multiple rounds of decoding with the combined decoder after receiving the retransmitted codeword. 
	Suppose $\mathbf{z}=(\mathbf{z}_1,\mathbf{z}_2,\mathbf{z}_3,\mathbf{z}_4)$ is a codeword of $\mathcal{C}_{I}$ or $\mathcal{C}_{II}$. 
	According to the above definition, $\phi_4(\mathbf{z})$ is a $4\times 23$ matrix.

%
%

	Let $\mathbf{z}^{\prime}=(\mathbf{z}_1,\mathbf{z}_2X^{-1},\mathbf{z}_3X^{-2},\mathbf{z}_4X^{-3})$, where $X$ is a cyclic shift matrix of order $n$. 
	In practice, there are many ways to implement Chase Combining, both ``soft" and ``hard", but in this paper we only use the most basic method, which is majority voting.  
	For a specific column $\phi_4(\mathbf{z}^{\prime})$, if the number of erroneous bits is less than half, the error is termed a \textit{partial column-error}; otherwise, it is classified as a \textit{majority column-error}. 
	According to the SBDH of \( \mathcal{C}_{I} \), if $\mathbf{z}$ is a codeword of $\mathcal{C}_{I}$, then the following types of error patterns in $\phi_4(\mathbf{z})$ and $\phi_4(\mathbf{z}^{\prime})$ can be corrected: 
	\begin{enumerate}
		\item Errors of no more than 3 bits in any single row of $\phi_4(\mathbf{z})$;
		\item Errors in no more than 5 columns across any two rows of $\phi_4(\mathbf{z})$;
		\item Errors in no more than 6 columns across any three rows of $\phi_4(\mathbf{z})$;
		\item Errors in no more than 7 columns across all four rows of $\phi_4(\mathbf{z})$;
		\item All partial column-error and no more than $3$ majority column-errors of $\phi_4(\mathbf{z}^{\prime})$.
	\end{enumerate}

	According to SBDH and SHDH  of \( \mathcal{C}_{II} \), if $\mathbf{z}$ is a codeword of $\mathcal{C}_{II}$, then the following types of error patterns in $\phi_4(\mathbf{z})$ can be corrected: 
	\begin{enumerate}
		\item Errors of no more than 3 bits in any single row of $\phi_4(\mathbf{z})$; 
		\item Errors in no more than 6 bits or 5 columns across any two rows of $\phi_4(\mathbf{z})$;
		\item Errors in no more than 11 bits or 7 columns across any three rows of $\phi_4(\mathbf{z})$;
		\item Errors in no more than 17 bits or 9 columns across all four rows of $\phi_4(\mathbf{z})$;  
	\end{enumerate}
	
	In contrast, the classical repetition Golay code $C_1$ has parameters $[92,12,28]_2$ and $4$-symbol Golay code $C_2$ has a minimum $4$-symbol distance $15$. 
	According to Grassl's codetable \cite{Grassltable}, the optimal $[92,12]_2$ linear code $C_3$ has minimum distance $40$.

	Here, $C_3$ is treated as an IR-linear code, which involves dividing it into four equal-length segments and transmitting them sequentially over four communication rounds. 

	We conduct simulations on these codes at an SNR of $-5$ dB, as shown in Figure \ref{fig_second_case}.
	It can be seen from (a) of Figure \ref{fig_second_case} that when $m \geq 2$, the FER of $b$-symbol code $C_2$ is lower than that of classical repetition code $C_1$. 
	When $m \geq 3$, the performance of Type-I GRC $\mathcal{C}_I$ is surpasses than that of the $b$-symbol code $C_2$. 
	As the decoding depth increases, the FER of $\mathcal{C}_I$ decreases, but the gain from further increasing the depth diminishes.

	Similarly, the FER of Type-II GRC decreases with increasing decoding depth, and the gain from decoding depth diminishes as \( m \) increases. Specifically, when \( m = 2 \), the FER of Type-II GRC with a decoding depth of 2 is lower than that of \( C_3 \). 
	For \(m = 3 \) and \(m = 4 \), $\mathcal{C}_{II}$ outperforms \(C_3 \) at decoding depths of $3$ and $4$, respectively.
	\hfill $\square$

	\begin{figure*}[ht]
		\centering
		\subfloat[]{\includegraphics[scale=0.22]{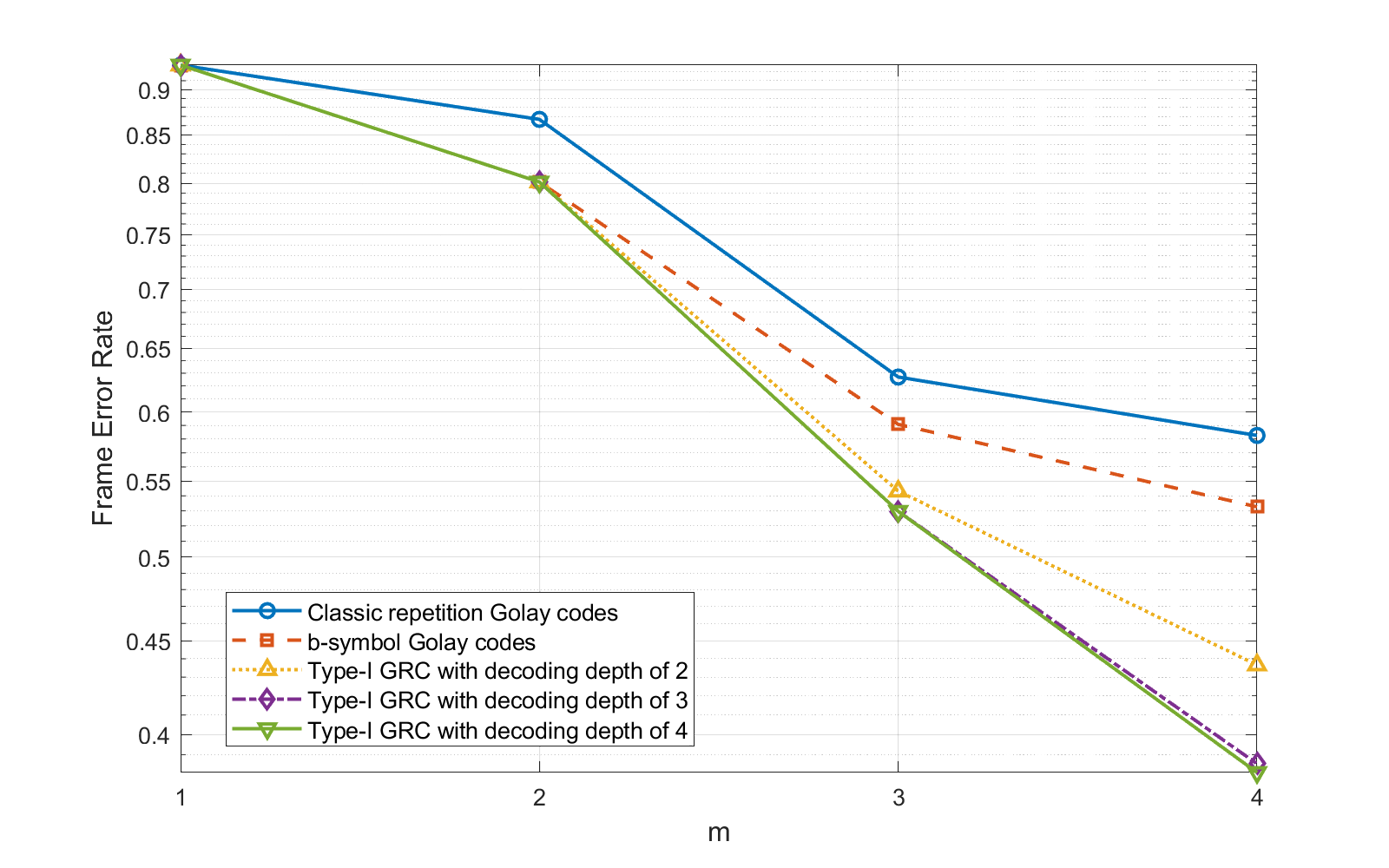}%
		}
		\subfloat[]{\includegraphics[scale=0.22]{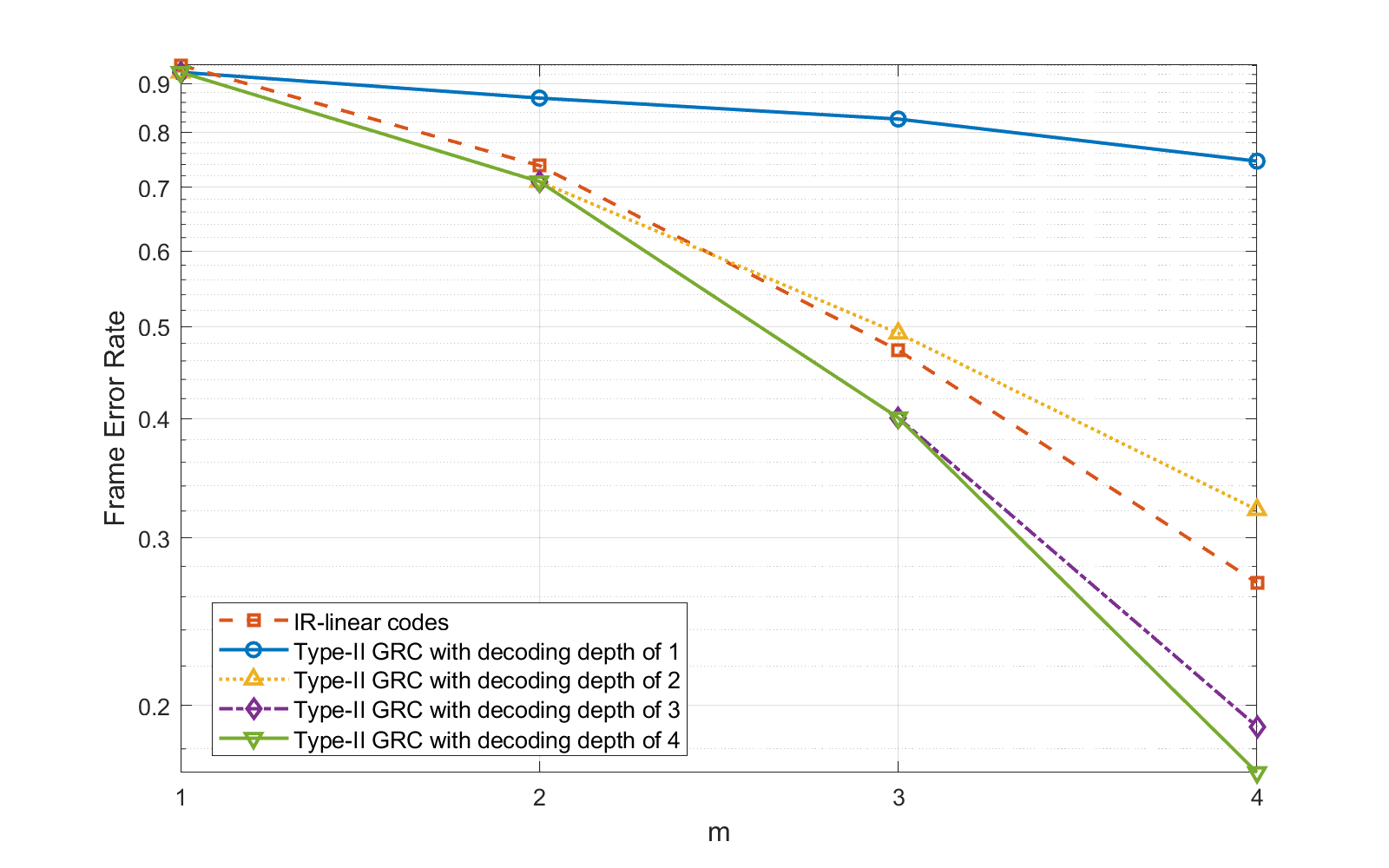}%
		}
		\caption{Comparison of FER between GRCs and Classical Repetition Codes under ML Decoding Algorithm under SNR=$-5$ dB. (a) Type-I GRC $\mathcal{C}_I$ with different decoding depths VS. classical repetition code $C_1$ and $b$-symbol code $C_2$.   
			(b) Type-II GRC $\mathcal{C}_{II}$ VS. IR-linear code \( C_3 \).}
		\label{fig_second_case}
	\end{figure*}
\end{example}

As shown by the theoretical analysis and simulation results above, well-designed GRCs exhibit superior FER performance compared to classical repetition codes, $b$-symbol codes, and IR-linear codes. Therefore, designing GRCs with suitable parameters is crucial. The subsequent sections present the theoretical bounds and construction methods for GRCs.

\section{Bounds and constructions of Type-I generalized repetition codes} \label{Sect IV}
 
In this section, we refine the Griesmer bound in Lemma \ref{Griesmer_Bound} for Type-I GRCs using finite geometry techniques. 
We then propose methods to construct Type-I regular GRCs with bounded SBDH by leveraging cyclic codes, extended cyclic codes, and quasi-cyclic codes.

\subsection{Upper bounds on Type-I generalized repetition codes}

\begin{definition}
	Suppose $\mathcal{C}$ is an $[(n,m),k]_q$ linear block metric code with generator matrix $\mathcal{G}_m=(G_0,G_1,\ldots,G_{m-1})$, where $G_i\in \mathbb{F}_q^{k\times n}$. 
	For $\mathbf{p}=(\alpha_0,\alpha_1,\ldots,\alpha_{m-1})^T\in PG(m-1,q)$, denote $ \mathbf{p}\circ \mathcal{G}_m=\sum_{i=0}^{m-1} \alpha_i G_i$. 
	Let $\mathfrak{X}_m(\mathcal{G}_m)$ denote the $k\times nN_m$ matrix determined by the horizontal join of $N_m$ different matrices $\mathbf{p}_0\circ \mathcal{G}_m,\mathbf{p}_1\circ \mathcal{G}_m,\ldots,\mathbf{p}_{N_m-1}\circ \mathcal{G}_m$,  where $\mathbf{p}_0,\mathbf{p}_1,\ldots,\mathbf{p}_{N_m-1}$ are all points of $PG(m-1,q)$.
	Let $\mathfrak{X}_m(\mathcal{C})$ be the $[nN_m,k]_q$ linear code generated by $\mathfrak{X}_m(\mathcal{G}_m)$.
\end{definition}

\begin{lemma}\label{expand_flats}
	If $\mathcal{C}$ is an $[(n,m),k,d_m]_q$ linear block metric code with block weight distribution $1+\mathsf{A}_{d_m}z^{d_m}+\mathsf{A}_{d_m+1}z^{d_m+1}+\cdots+\mathsf{A}_nz^n$, then $\mathfrak{X}_m(\mathcal{C})$ is an  $[nN_{m},k, q^{m-1}d_m]_q$ linear code with weight distribution $1+\mathsf{A}_{d_m}z^{q^{m-1}d_m}+\mathsf{A}_{d_m+1}z^{q^{m-1}(d_m+1)}+\cdots+\mathsf{A}_nz^{q^{m-1}n}$.
\end{lemma}
\begin{proof}
	Let $H$ be a hyperplane of $PG(k-1,q)$ and let $\mathcal{F}_i$ be an  $(m-1)$-flat of $\mathcal{C}$. 
	If $\mathcal{F}_i$ lies in $H$, then $\dim(H\cap \mathcal{F}_i)=m$. Otherwise, $\dim(H\cap \mathcal{F}_i)=m-1$.
	If $H$ contains $n-w$ of $(m-1)$-flats of $\mathcal{C}$, where $w\ge d$, then $H$ contains $(n-w)N_{m}+wN_m=nN_m-wq^{m-1}$ points of $\mathfrak{X}_m(\mathcal{C})$.
	Therefore, $\mathfrak{X}_m(\mathcal{C})$ has parameters $[nN_{m},k, q^{m-1}d_m]_q$, and a weight distribution is similar to $\mathcal{C}$, but the weights of $\mathfrak{X}_m(\mathcal{C})$ are magnified by a factor of $q^{m-1}$. 
\end{proof}

Since Type-I GRCs possess a particular algebraic and geometric structure, the Griesmer bound in the Lemma \ref{Griesmer_Bound} for Type-I GRCs can be improved. 
\begin{theorem}\label{non-existence of $b$-symbol codes}
	Let $N_m=\frac{q^m-1}{q-1}$. 
	Assume that $s\equiv \left\lceil \frac{d_m}{q^{k-m}} \right\rceil \pmod{N_{m}}$ and  $n^*\equiv n\pmod{N_{k}}$.
	If $\mathcal{C}$ is an $[(n,m),k,d_m]_q$ non-degenerate Type-I GRC, then the following holds:
	
	(1) If $m> s$, then
	\begin{equation}
		nN_m\ge g_{q,m}(k,d_m)+m-s.
	\end{equation}
	
	(2) If $s=m$ and $mn^*> N_k$, then 
	\begin{equation}\label{E_s=b}
		nN_m\ge g_{q,m}(k,d_m)+1.
	\end{equation}
\end{theorem}
\begin{proof}
	(1) By Lemma \ref{expand_flats}, $\mathfrak{X}_m(\mathcal{C})$ is a linear code with parameters $[nN_m,k,q^{m-1}d_m]_q$.  
	Let $s=\epsilon N_{m}+\zeta$, where $ 1\le \zeta \le m-1$. 
	For $\epsilon=0$, there exist at least  $n$ points lying on the $m$ $(m-1)$-flats of $\mathcal{C}$, respectively. Then the generator matrix of $\mathfrak{X}_m(\mathcal{C})$ is 
	\begin{equation}
		G=\left(\begin{array}{c|c}
			\mathbf{1}_m
			&
			\begin{array}{r}
				\cdots 
			\end{array}\\ \hline
			\mathbf{0}
			&
			\begin{array}{c}
				G^{\prime}  \\
			\end{array} 
			
		\end{array}\right).
	\end{equation}

	Let $C^{\prime}$ be the code generated by $G^{\prime}$. Then, $C^{\prime}$ has parameters $[nN_{m}-m,k-1,q^{m-1}d_m]_q$.
	Since $g_{q,m}(k-1,d_m)=\sum\limits_{i=0}^{k-2}\left\lceil\frac{d_m}{q^{i-m+1}}\right\rceil=g_{q,m}(k,d_m)-\left\lceil\frac{d_m}{q^{k-m}}\right\rceil=g_{q,m}(k,d_m)-s\le nN_{m}-m$, we have $g_{q,m}(k,d_m)+m-s\le nN_{m}$.

	If $\epsilon,\zeta\ge 1$, then we have $nN_m\ge g_q(k,q^{m-1}d_m)> \epsilon N_mN_k$.  
	So $n>\epsilon N_k$.  
	Therefore, there exists a point $\mathbf{p}$ that repeats at least $m(\epsilon+1)$ times in the generator matrix of $\mathcal{C}$.  
	Since $\frac{n(N_m-m)}{N_k}>\epsilon(N_m-m)$, the same point $\mathbf{p}$ repeats at least $\epsilon(N_m-m)+m(\epsilon+1)=\epsilon N_m+m$ times in $\mathfrak{X}_m(\mathcal{C})$.  
	Similarly, we can deduce that $nN_{m}\ge g_{q,m}(k,d_m)+m-s$.

	(2) 
	The $nN_m$ points (which may include repetitions) of $\mathfrak{X}_m(\mathcal{C})$ can be divided into two parts: the first part is $mn$ points of $\mathcal{C}$, and the second part is the remaining $(N_m-m)n$ points.
	
	Suppose $n=\epsilon N_k+n^*$. 
	As $m\ge2$, $N_m\ge 2m-1$.
	Since $mn^*>N_k$, we have $n^*+(N_m-m)n^*\ge mn^*>N_k$.
	If $\epsilon=0$, then the first and second parts have a nonempty intersection.
	Thus, $\mathfrak{X}_m(\mathcal{C})$ has a point $\mathbf{p}$ repeats at least $m+1$ times.

	If $\epsilon\ge 1$, then one can easily see that $\mathfrak{X}_m(\mathcal{C})$ has a point $\mathbf{p}$ that repeats at least $\epsilon N_m+m+1$ times.
	Therefore, by a method similar to that in (1), we have that Equation (\ref{E_s=b}) holds.
\end{proof}

	 Baumert and McEliece \cite{baumert1973note} proved that optimal linear codes always exist that can achieve the Griesmer bound, given a sufficiently large code length. 
	Nonetheless, Theorem \ref{non-existence of $b$-symbol codes} suggests that even if the code length is sufficiently large, Type-I GRCs do not necessarily reach the Griesmer bound in Lemma \ref{Griesmer_Bound}.

\begin{example}
	In \cite{Luo2024GriesmerBA}, Luo et al. constructed two classes of $b$-symbol optimal codes with parameters $(q^k-1,k+1,q^k-q^{k-b}-1)_q^b$ and $(q^k,k+1,q^k-q^{k-b})_q^b$, respectively, where $b\le k$. 
	They are equivalent to $[(q^k-1,b),k+1,q^k-q^{k-b}-1]_q$ and $[(q^k,b),k+1,q^k-q^{k-b}]_q$ Type-I GRCs. 
	However, according to the Griesmer bound in Lemma \ref{Griesmer_Bound}, there may exist $[(q^k-1,b),k+b,q^k-q^{k-b}-1]_q$ and $[(q^k,b),k+b,q^k-q^{k-b}]_q$ Type-I GRCs, those two codes have the same distances but higher dimensions. 
	Theorem \ref{non-existence of $b$-symbol codes} implies that no Type-I GRCs exist with parameters $[(q^k-1,2),k+2,q^k-q^{k-2}-1]_q$ or $[(q^k,2),k+2,q^k-q^{k-2}]_q$.
%
	\hfill $\square$ \end{example}

According to Definition \ref{D_repetition code}, the construction problem of Type-I GRCs involves two key aspects. Firstly, it is to construct a special linear code $C$ under the Hamming metric. Secondly, it involves designing appropriate permutation matrix $A_1,\ldots,A_{m-1}$. 
Based on this, for practical applications, we focus on constructing Type-I regular GRCs. 
Let $C$ be an $[n,k]_q$ linear code with generator matrix $G$. Then, by Definition \ref{D_repetition code}, an $[(n,m),k]_q$ Type-I regular GRC $\mathcal{C}$ based on $C$ has a generator matrix
\begin{equation}\label{Eq_regular}
	\mathcal{G}=(G,GA,\ldots,GA^{m-1}),
\end{equation}
where $A\in Ps(n,\mathbb{F}_q)$. 
Suppose $\sigma$ is the corresponding permutation of $A$. For an integer $m\le n$, we define the Type-I regular GRC of $C$ under permutation $\sigma$ as $C_{\sigma,m}=C\times \sigma(C)\times \cdots \times\sigma^{m-1}(C)$. 
For a codeword $\mathbf{c}$ of $C$, we define $\sigma^{(m)}(\mathbf{c})=(\mathbf{c},\sigma(\mathbf{c}),\ldots,\sigma^{m-1}(\mathbf{c}))$. 
The following theorem establishes a upper bound for Type-I regular GRC of $C_{\sigma,m}$. 

%

\begin{theorem}\label{Theo_MDS_bound}
	Let $C$ be an $[n,k]_q$ full length linear code with generator matrix $G$ and let $\sigma$ be a permutation with maximum cyclic length $\jmath \ge k$.  
	Then, for $m\le k $, the SBDH $ (\mathbf{d}_r(C_{\sigma,m}) )_{r = 1}^m$ of Type-I regular GRC $C_{\sigma,m}$ satisfies 
	\begin{equation}\label{Eq_MDS_bound}
		\mathbf{d}_r(C_{\sigma,m})\le n-k+r.
	\end{equation} 
\end{theorem}
\begin{proof}
	By the properties of permutations, \( \sigma \) can be decomposed into disjoint cycles. 
	We can express the permutation \( \sigma \) as $\sigma = \sigma_1 \sigma_2 \cdots \sigma_s$, where each \( \sigma_i \) denotes a disjoint cycle. 
	Without loss of generality, assume that the coordinates involved in the permutations \( \sigma_1, \sigma_2, \ldots, \sigma_s \) are arranged in increasing order and the length of $\sigma_1$ is $\jmath$.
	Then, for a codeword \( \mathbf{c} \) of \( C \), we can partition it into \( s \) subvectors according to the coordinates involved in the permutation \( \sigma_i \), and 
	$\mathbf{c} = (\mathbf{c}_1, \mathbf{c}_2, \ldots, \mathbf{c}_s)$.
	
	Then, the array form of codeword $\sigma^{(m)}(\mathbf{c})$ of $C_{\sigma,m}$ can be denoted by 
	\begin{equation}
		\phi_m(\sigma^{(m)}(\mathbf{c}))=\left(\begin{array}{cccc}
			\mathbf{c}_1& \mathbf{c}_2& \cdots& \mathbf{c}_s\\
			\sigma_1(\mathbf{c}_1)& \sigma_2(\mathbf{c}_2)& \cdots& \sigma_s(\mathbf{c}_s)\\
			\vdots &\vdots&\ddots&\vdots\\
			\sigma_1^{m-1}(\mathbf{c}_1)& \sigma_2^{m-1}(\mathbf{c}_2)& \cdots& \sigma_s^{m-1}(\mathbf{c}_s)\\
		\end{array}\right).
	\end{equation}
	 Since each part of \( \phi_m(\sigma^{(m)}(\mathbf{c})) \), when expanded individually into a row vector, can be regarded as a codeword of a \( b \)-symbol code, the Type-I regular GRC \( C_{\sigma,m} \) can be decomposed into \( s \) \( b \)-symbol codes.
	 Due to the length $\jmath$ of $\sigma_1$ satisfy $\jmath\ge k\ge m$ and $C$ is full length, 
	 the first \( \jmath \) coordinates of \( C \) can generate a $[(\jmath,m),k,d_m']$  \( m \)-symbol code \( C_b \) under the action of \( \sigma_1 \).
	Thus, we have $\mathbf{d}_m(C_{\sigma,m})\le n-\jmath+d_m'$.
	By the Singleton bound on $b$-symbol codes in Lemma \ref{Singleton_bound_b}, we get $d_m'\le \jmath-k+m$ and $\mathbf{d}_m(C_{\sigma,m})\le n-k+m.$ 
 	A similar argument can be used to show that \( C_{\sigma,m}^{[r]} \) satisfies \( \mathbf{d}_r(C_{\sigma,m}^{[r]}) \leq n - k + r \).
 	Since $\mathbf{d}_r(C_{\sigma,m}) \le \mathbf{d}_r(C_{\sigma,m}^{[r]}) \le n-k+r$, we have $\mathbf{d}_r(C_{\sigma,m})\le n-k+r.$
	This completes the proof.
\end{proof}

\subsection{Constructions of Type-I regular generalized repetition codes from cyclic and quasi-cyclic codes}
It is well known that quasi-cyclic codes are an efficient generalization of cyclic codes and are asymptotically good \cite{kasami1974gilbert,Ling2003}. Numerous computational results \cite{Grassltable} and specialized databases \cite{Chentable} are available for quasi-cyclic codes. Thus, we consider leveraging cyclic and quasi-cyclic codes to construct Type-I regular GRCs.

As shown in the Section \ref{Sec II}, replacing $A$ with cyclic shift matrix $X$ in Equation (\ref{Eq_regular}) makes $\mathrm{d}_m(C_{\pi_\ell,m})$ be the minimum $m$-symbol distance of $C$.  
Extensive research has been conducted on constructing linear codes with larger minimum $b$-symbol distances; see \cite{chee2013maximum,kai2015construction,chen2017constacyclic,Ding2018MaximumDS,zhu2022complete}. 
However, as illustrated in Figure \ref{FER_of_b_symbol_Golay_codes}, focusing solely on the $b$-symbol distance of $C$ may reduce coding gain as $b$ increases.  
Therefore, unlike prior work, this paper emphasizes the SBDH of Type-I regular GRC $C_{\pi_\ell,m}$, not just $\mathrm{d}_m(C_{\pi_\ell,m})$.

The following theorem provides a lower bound on the SBDH of Type-I regular GRC $C_{\pi,m}$ from cyclic codes. 

\begin{theorem}\label{Theo_SBDH_lower_bounds}
		Let $C$ be an $[n,k,d]_q$ cyclic code with generator polynomial $g(x)$. 
 If $m\le \kappa(g(x))$, then the SBDH $ (\mathbf{d}_r(C_{\pi,m}) )_{r = 1}^m$ of Type-I regular GRC $C_{\pi,m}$ satisfies 
 \begin{equation} 
 	\mathbf{d}_r(C_{\pi,m})\ge g_q(r,d).
 \end{equation}    
\end{theorem}
\begin{proof}
	Suppose $\mathbf{c}$ is a codeword of $C$ and $c(x)$ is the coefficient polynomial of $\mathbf{c}$.   
	Let $\mathbf{c}_I=\pi^{(m)}(\mathbf{c})$.  
	Let $\Delta(x)=\gcd(x^n-1,c(x))$. 
	Then, we have $g(x) \mid \Delta(x)$ and $c(x)=a(x)\Delta(x)$, where $a(x)\in \mathbb{R}_{q,n}$ with degree less than $n-\deg(\Delta(x))$. 
	
	According to the definition of $\kappa(g(x))$, if $\mathrm{w}(\mathbf{c})<n$, then we have $\deg(\Delta(x))\le n-\kappa(g(x))$. 
	By the theory of cyclic codes \cite{huffman2010fundamentals}, it follows that $c(x),xc(x),\ldots,x^{m-1}c(x)$ are linearly independent polynomials. 
	Thus, we get that, for $T\in \binom{[m]}{r}$,  $\phi_r(\mathbf{c}_I^T)$ generates an $[n,\mathrm{rank}(\phi_r(\mathbf{c}_I^T))]_q$ linear code $C_r$.  
	By the Griesmer bound,  we have $n(C_r)=w_r(\mathbf{c}_I^T)\ge g_q(r,d)$. 
\end{proof}

The following example demonstrates the validity of Theorem \ref{Theo_SBDH_lower_bounds}. 
\begin{example}
	Let $q=2$ and $n=15$. Consider $g(x)=x^{11} + x^{10} + x^9 + x^8 + x^6 + x^4 + x^3 + 1$. Then, $g(x)$ generates a $4$-dimensional cyclic Simplex code $C$. 
	Since the parity check polynomial of $C$ is an irreducible polynomial, we have $\kappa (g(x))=4$. 
	By Theorem \ref{Theo_SBDH_lower_bounds}, $\mathbf{g}(x)=(g(x),xg(x),x^2g(x),x^3g(x))$ generates a $[(15,4),4]_2$ Type-I regular GRC with SBDH $(\mathbf{d}_1(\mathcal{C})=8, \mathbf{d}_2(\mathcal{C})\ge 12,\mathbf{d}_3(\mathcal{C})=14,\mathbf{d}_4(\mathcal{C})=15)$. 
	
	However, if $m>\kappa (g(x))$, the lower bounds in Theorem \ref{Theo_SBDH_lower_bounds} may not hold. 
	Let $C_1$ denote the dual code of $C$ and $h(x)=\frac{x^{15}-1}{g(x)}$. 
	Then, $C_1$ is the $[15,11,3]_2$ cyclic Hamming code.
	Since $g(x)=(x+1)(x^2 + x + 1)(x^4 + x + 1)(x^4 + x^3 + x^2 + x + 1)$, we have $\kappa (h(x))=2$. 
	It is easy to check that $\mathbf{h}(x)=(h(x),xh(x),x^2h(x),x^3h(x),x^4h(x),x^5h(x))$ generates a $[(15,6),11]_2$ Type-I GRC $\mathcal{C}_1$ with $\mathbf{d}_2(\mathcal{C}_1)=3$.  \hfill $\square$
\end{example}

Extended cyclic codes, derived from classical cyclic codes, play a vital role in modern communication and storage systems.  
Many well-known code classes are extended cyclic codes, including first-order Reed-Muller codes, extended BCH codes, and extended Golay codes \cite{macwilliams1977theory}.
In the following theorem, we give a method to construct Type-I GRC with bounded the SBDH of Type-I regular GRC from extended cyclic codes.  

\begin{theorem}\label{Theo_extended_cyclic}
	Let $\pi'=(1,2,\ldots,n)(n+1)$ and let $C$ be an $[n,k,d]_q$ cyclic code with generator polynomial $g(x)$.   
	If $\mathrm{d}(\widehat{C})=d+1$ and $ m\le \kappa(g(x))$, 
	then the SBDH $ (\mathbf{d}_r(\widehat{C}_{\pi',m}) )_{r = 1}^m$ of Type-I regular GRC $\widehat{C}_{\pi',m}$ satisfies 
	\begin{equation} 
		\mathbf{d}_r(\widehat{C}_{\pi',m})\ge g_q(r,d+1).
	\end{equation}    
	 
\end{theorem}
\begin{proof}
	Suppose $\mathbf{c}$ is a nonzero codeword of $C$ with weight less than $n$ and $\mathbf{c}_I=\pi'^{(m)}(\mathbf{c})$.  
	According to Theorem \ref{Theo_SBDH_lower_bounds}, for all $T\in \binom{[m]}{r}$, $g_q(r,d)\le \mathrm{w}_r(\mathbf{c}_I^T)$. 
	Let $C_r$ be a linear code generated by $\phi_r(\mathbf{c}_I^T)$. 
	It follows that $C_r$ is a subcode of $C$. 
	Since $\mathrm{d}(\widehat{C})=d+1$, for all minimum weight codewords $\mathbf{c}$ of $C$, we have $\widehat{\mathbf{c}}=(\mathbf{c},\alpha)$, where $\alpha\in \mathbb{F}_q^*$. 
	This implies that if $\mathrm{d}(C_r)=d$, then $\mathrm{d}(\widehat{C}_r)=d+1$. 
	Hence, $\mathrm{w}_r(\mathbf{c}_I^T)=n(\widehat{C}_r)\ge g_q(r,d+1)$. 
	This completes the proof.
\end{proof}
%



\begin{example}
	According to Chapter 1.9 of \cite{huffman2010fundamentals}, there exist only two Golay codes: the $[23,11,7]_2$ code $\mathcal{G}_1$ and the $[11,6,5]_3$ code $\mathcal{G}_2$, both of which are cyclic.
	Their generator polynomials are $g_1(x)=x^{11} + x^9 + x^7 + x^6 + x^5 +  x + 1$ and $g_2(x)=x^5 + 2x^3 + x^2 + 2x + 2$. 
	Since $\kappa (g_1(x))=11$ and $\kappa (g_2(x))=5$, by Theorem \ref{Theo_SBDH_lower_bounds}, we directly obtain $[(23,11),12]_2$ and $[(11,5),6]_3$ Type-I regular GRCs with SBDH  $(7,11,13,15,16,17,18,19,20,21,22)$ and $(5,7,8,9,10)$, respectively. 
	
	Moreover, by Theorem \ref{Theo_extended_cyclic}, one can use the extended binary and ternary Golay codes with parameters $[24,12,8]_2$ and $[12,6,6]_3$ to deduce $[(24,11),12]_2$ and $[(12,5),6]_3$ Type-I GRCs with SBDH  $(8,12,14,16,17,18,19,20,21,22,23)$ and $(6,8,9,10,11)$, respectively.   \hfill $\square$
\end{example}

Let \(\pi_m=(1,\ldots,n)(n+1,\ldots,2n)\cdots((m-1)n+1,\ldots,m n )\) denote the quasi-cyclic permutation, with its corresponding permutation matrix denoted by  \(X_{m}\).

\begin{theorem}\label{Theo_QC_Type-I}
		Suppose $C$ is an $[\ell n,k,d]_q$ quasi-cyclic code with generator  $\mathbf{g}(x)=(g(x),g(x)f_1(x),\ldots,g(x)f_{\ell-1}(x))$, where $g(x)$ and $f_i(x)$ are polynomials in $\mathbb{R}_{q,n}$ such that $g(x)\mid(x^n-1)$ and $\gcd(f_1(x),\ldots,f_{\ell-1}(x),x^n-1)=1$. 
	Let $\delta=\#\left\{ f_i(x)\,\bigg|\, \gcd\left( f_i(x),\frac{x^n-1}{g(x)} \right)\ne 1, i \in [\ell-1] \right \}$. 
	 If $m\le \kappa(g(x))$, then the SBDH $ (\mathbf{d}_r(C_{\pi_\ell,m}) )_{r = 1}^m$ of Type-I regular GRC $C_{\pi_\ell,m}$ satisfies 
	 \begin{equation}
	 	 \begin{array}{cl}
	 		g_q(r,d)\le \mathbf{d}_r(C_{\pi_\ell,m}) \le \ell n-k+r,& \delta=0,\\
	 		\min \{ g_q(r,d),(\ell-\delta)n \}\le \mathbf{d}_r(C_{\pi_\ell,m}) \le \min \left\{ \ell n-k+r,\left(\ell-1\right) n \right\}, & \delta >0.
	 	\end{array} 
	 \end{equation} 
\end{theorem}
\begin{proof}
	Write nonzero codeword $\mathbf{c}$ of $C$ as $(\mathbf{c}_1,\mathbf{c}_2,\ldots,\mathbf{c}_m)$, where $\mathbf{c}_i\in \mathbb{F}_q^n$. 
	$\mathbf{c}_i$ can be regard as a codeword of a cyclic code with polynomial $g(x)f_{i-1}(x)$, where $i\ge 2$. 
	If $\gcd\left( f_i(x),\frac{x^n-1}{g(x)} \right)\ne 1$, then $\mathbf{c}_i\ne \mathbf{0}$. 
	Thus, we have that for all codeword $\mathbf{c}$ of $C$, $\mathbf{c}$ has at most $\delta$ zero sub-vectors $\mathbf{c}_i$. 
		Let $\mathbf{c}_I=\pi^{(m)}_\ell (\mathbf{c})$. Moreover, according to the proof of Theorem \ref{Theo_SBDH_lower_bounds}, for all $T\in \binom{[m]}{r}$, $\mathrm{rank}(\phi_r(\mathbf{c}_I^T))=r$.
	If all sub-vectors of $\mathbf{c}$ are nonzero vectors, then we have $\mathrm{w}_r(\phi_r(\mathbf{c}_I^T))\ge g_q(r,d)$. 
	If $\mathbf{c}$ has $z$ zero sub-vectors, then $\mathrm{w}_r(\phi_r(\mathbf{c}_I^T))$ can only be bounded by $(\ell-z)n$, where $z\le \delta$. 
	Hence, we have $\mathbf{d}_r(C_{\pi_\ell,m})\ge \min \{ g_q(r,d),(\ell-\delta)n \}$. 
	However, if $\delta> 0$, then $C$ has at least one codeword $\mathbf{c}$ with one sub-vector $\mathbf{c}_i=\mathbf{0}$. It follows that  $\mathrm{w}_r(\mathbf{c}_I^T)\le (m-1)n$. 
\end{proof}

In Chen's quasi-cyclic code table \cite{Chentable}, there are numerous good quasi-cyclic codes. 
By means of Theorem \ref{Theo_QC_Type-I}, one can directly obtain a large number of Type-I regular GRCs with good parameters from known quasi-cyclic codes.

\section{Construction of Type-II generalized repetition codes from linear codes}	\label{Sect V}

The previous section presents methods for constructing Type-I regular GRCs from (extended) cyclic and quasi-cyclic codes with bounded SBDH.
However, most linear codes are neither cyclic nor quasi-cyclic \cite{Grassltable}, and Type-I GRCs derived from such codes must satisfy $m \le \kappa(g(x))$.
In particular, Type-I regular GRCs are bounded by the upper bound given in Theorem \ref{Theo_MDS_bound}.
Therefore, in this section, we propose a method to construct Type-II GRCs with bounded SBDH from general linear codes.
Moreover, we show that Type-II GRCs can exceed the bound in Theorem \ref{Theo_MDS_bound}.
We also propose a method to use quasi-cyclic codes to construct Type-II GRCs and obtain a large number of Type-II GRCs with good SBDH and SHDH by computer search. 

\subsection{Construction of Type-II generalized repetition codes with bounded SBDH from general linear codes}

 We first establish an upper bound on the minimum sub-Hamming distance for Type-II GRCs.
\begin{theorem}\label{Theorem_Type-II_Bound}
	Suppose $\mathcal{C}$ is an $[(n,m),k,d_m]_q$ non-degenerate Type-II GRC. 
	Then,  
	\begin{equation}
		\underline{\mathbf{d}}_m(\mathcal{C}) \le q^{m-1}d_m -\left(\frac{q^m-1}{q-1}-m \right)\underline{\mathbf{d}}_1(\mathcal{C})
	\end{equation}
\end{theorem}
\begin{proof}
	Let $\mathcal{C}$ have the generator matrix $\mathcal{G}_m = (G, B_1G, \ldots, B_{m-1}G)$.
	By Lemma \ref{expand_flats}, $\mathfrak{X}_m(\mathcal{C})$ is an  $[nN_{m},k, q^{m-1}d_m]_q$ linear code whose generator matrix is $\mathfrak{X}_m(\mathcal{G}_m)=(\mathbf{p}_0\circ \mathcal{G}_m,\mathbf{p}_1\circ \mathcal{G}_m,\ldots,\mathbf{p}_{N_m-1}\circ \mathcal{G}_m)$,  where $\mathbf{p}_0,\mathbf{p}_1,\ldots,\mathbf{p}_{N_m-1}$ are all points of $PG(m-1,q)$.
	
	Without loss of generality, assume the first $m$ points $\mathbf{p}_0, \mathbf{p}_1, \ldots, \mathbf{p}_{m-1}$ correspond to the standard basis vectors of $\mathbb{F}_q^m$, i.e., these points form an identical matrix. 
	Then,  $\mathfrak{X}_m(\mathcal{G}_m)=(\mathcal{G}_m,\mathbf{p}_m\circ \mathcal{G}_m,\ldots,\mathbf{p}_{N_m-1}\circ \mathcal{G}_m)$.  
	The code obtained by removing the blocks $(\mathbf{p}_m \circ \mathcal{G}_m, \ldots, \mathbf{p}_{N_m-1} \circ \mathcal{G}_m)$ from $\mathfrak{X}_m(\mathcal{G}_m)$ has a minimum Hamming distance at most $q^{m-1}d_m - (N_m - m)\underline{\mathbf{d}}_1(\mathcal{C})$.
	It follows that 
	$\underline{\mathbf{d}}_m(\mathcal{C}) \le q^{m-1}d_m -(N_m-m)\underline{\mathbf{d}}_1(\mathcal{C})$. 
	Since $N_m = \frac{q^m-1}{q-1}$, the proof is complete.
\end{proof}

When using only a single code for communication, one always desires \( \underline{\mathbf{d}}_1(\mathcal{C}) \) to be as large as possible, so that the receiver can perform decoding with low complexity. However, this theorem reveals that there exists a certain trade-off among \( \mathbf{d}_m(\mathcal{C}) \), \(\underline{\mathbf{d}}_m(\mathcal{C}) \), and \( \underline{\mathbf{d}}_1(\mathcal{C}) \); that is, it may not be possible to simultaneously achieve large values of \( \mathbf{d}_m(\mathcal{C}) \), \(\underline{\mathbf{d}}_m(\mathcal{C}) \), and \( \underline{\mathbf{d}}_1(\mathcal{C}) \).
\begin{example}
	Let $q=2$ and $n=15$. Let $g(x)=x^9 + x^7 + x^6 + x^3 + x^2 + 1$ and $f(x)=x^3 + x^2 + 1$, where $g(x)$ is a generator polynomial of $[15,6,6]_2$ cyclic code. 
	Then, it is easy to check that $\mathbf{g}(x)=(g(x),f(x)g(x))$ generates a $[30,6,12]_{2}$ quasi-cyclic code $\mathcal{C}$.  
 
	One can use Magma~\cite{bosma1997magma} to compute that $\mathbf{d}_1(\mathcal{C})=\underline{\mathbf{d}}_1(\mathcal{C})=6$, $\mathbf{d}_2(\mathcal{C})=11$, and $\underline{\mathbf{d}}_2(\mathcal{C})=12$. 
	By Theorem \ref{Theo_MDS_bound}, $\mathcal{C}$ has optimal minimum block distance $11$.  
	In contrast, according to \cite{Grassltable}, the minimum Hamming distance of the optimal $[30,6]_2$ linear code is $14$.
	However, according to Theorem \ref{Theorem_Type-II_Bound}, when $\mathbf{d}_2(\mathcal{C})=11$, \( \underline{\mathbf{d}}_2(\mathcal{C}) \) can be at most $12$. This reveals the constraint relationship between the minimum block distance and the Hamming distance for Type-II GRCs.	
	\hfill $\square$
\end{example}
 
Next, we present a method for constructing Type-II GRCs with bounded \(\mathbf{d}_m(\mathcal{C})\) using general linear codes. 
The following lemma presents a property of the characteristic polynomial of a matrix.

\begin{lemma}[Characteristic polynomial,\cite{2010GeneralizationOS}]
	For any monic polynomial $f(x)$ of degree $k$ over $\mathbb{F}_q$, there exists a matrix $B\in GL_k(\mathbb{F}_q)$ such that the characteristic polynomial of $B$ is $f(x)$.
\end{lemma}

The existence of an irreducible polynomial of any degree $k$ over $\mathbb{F}_q$ is guaranteed by the properties of finite fields. One straightforward method is to take the generator polynomial of the $k$-th order cyclic Hamming code over $\mathbb{F}_q$.

\begin{theorem}\label{nonCyclic_mound} 
	Let $B\in GL_k(\mathbb{F}_q)$, whose characteristic polynomial is irreducible over $\mathbb{F}_q$.
	Let $C$ be an $[n,k]_q$ linear code with generator matrix $G$.
	Then, $\mathcal{G}_k=(G,BG,\ldots, B^{k-1}G)$ generates an $[(n,k),k]_{q}$ Type-II GRC with SBDH $\mathcal{C}$ has SBDH $ (\mathbf{d}_r(\mathcal{C})\ge g_q(r,d) )_{r = 1}^k$.	
\end{theorem}
\begin{proof}
	Let $\mathbf{u}$ be a vector in $\mathbb{F}_q^k$ and define  $\mathbf{c}_{k}=\mathbf{u}(G,BG,\ldots, B^{k-1}G)$.
	Suppose the eigenvalues of $B$ are $\lambda_1,\lambda_2,\ldots,\lambda_k$, where $\lambda_i\notin \mathbb{F}_q$.
	Let $D=\mathrm{diag}(\lambda_1,\lambda_2,\ldots,\lambda_k)$ be eigenvalue matrix of $B$. 
	Then, according to the matrix theory, $B^i$ can be written in the form of  $VD^iV^{-1}$, where $V\in GL_k(\mathbb{F}_q)$.
	Thus, $\mathbf{c}_{k}=\mathbf{u}V(I,D,\ldots, D^{k-1})V^{-1}G$.

	Suppose the characteristic polynomial of $B$ is $f_B(\lambda)=1+\zeta \lambda+\cdots+\zeta_k \lambda^k$. 
	Let $f_{\mathbf{\alpha}}(\lambda)=\alpha_0+\alpha_1\lambda+\cdots+\alpha_{m-1}\lambda^{m-1}$, where $m\in [k]$ and $(\alpha_0,\alpha_1,\ldots,\alpha_{m-1})^T\in PG(m-1,q)$. 
	Then, write $f_{\mathbf{\alpha}}(\lambda)$ as $f_{\mathbf{\alpha}}(\lambda)=f_{\mathbf{\alpha}}^{(1)}(\lambda) f_{\mathbf{\alpha}}^{(2)}(\lambda)$, where all roots of $f_{\mathbf{\alpha}}^{(1)}(\lambda)$ lie in $\mathbb{F}_q$ and $f_{\mathbf{\alpha}}^{(2)}(\lambda)$ has no root in $\mathbb{F}_q$.
	Let $E=\mathrm{lcm}(k,\deg(f_{\mathbf{\alpha}}^{(2)}(\lambda)))$. Then, both  $f_B(\lambda)$ and $f_{\mathbf{\alpha}}^{(2)}(\lambda)$ can be completely decomposed over $\mathbb{F}_{q^E}$.
	As $\deg(f_{\mathbf{\alpha}}^{(2)}(\lambda))\le k-1$ and $f_B(\lambda)$ is an irreducible polynomial over $\mathbb{F}_q$, 
	all the roots of $f_B(\lambda)$ are not conjugate to the roots of $f_{\mathbf{\alpha}}^{(2)}(\lambda)$. 
	Therefore, for all $(\alpha_0,\alpha_1,\ldots,\alpha_{k-1})^T\in PG(k-1,q)$, we have $f_{\mathbf{\alpha}}(\lambda_i)\ne0$. 
	It follows that $\mathrm{rank}(\alpha_0I_k+\alpha_1D+\cdots+\alpha_{k-1}D^{k-1})=k$ and 
	$\mathrm{rank}(\phi_k(\mathbf{c}_{k}))=k$.
	
	Let $C_{\mathbf{c}_{k}}$ be the $[n,k]$ linear code spanned by the rows of $\phi_k(\mathbf{c}_{k})$. Since row vector $\mathbf{u}B^iG$ belongs to $C$, the code $C_{\mathbf{c}_{k}}$ is a subcode of $C$ and hence has minimum distance at least $d$.
%
	Due to $\mathrm{w}_k(\mathbf{c}_{k})=|\mathrm{Supp}(\phi_k(\mathbf{c}_{k}))|$, by Lemma \ref{Griesmer_Bound}, we have $\mathrm{w}_k(\mathbf{c}_{k})\ge  g_q(k,d)$. 
	Similarly, as arbitrary $r$ of $\mathbf{u}G, \mathbf{u}BG, \ldots, \mathbf{u}B^{k-1}G$ are linearly independent, for all $T\in \binom{[k]}{r}$, we have $\mathrm{w}_r(\mathbf{c}_{k}^T)\ge  g_q(r,d)$. 
	Thus, $\mathcal{C}$ has SBDH $ (\mathbf{d}_r(\mathcal{C})\ge g_q(r,d) )_{r = 1}^k$.	
\end{proof}





Now, we show that the lower bound in Theorem \ref{nonCyclic_mound} is tight in many cases and that it can help us to directly construct SBDH-optimal Type-II GRCs from existing linear codes.

\begin{example}
	In \cite{Grassltable}, the $[11,4,5]_2$ linear code $C$ has a generator matrix $G$. 
	Following Theorem \ref{nonCyclic_mound}, we select a matrix $B \in GL_4(\mathbb{F}_2)$ with characteristic polynomial $f(x) = x^4 + x + 1$, defined as follows:
	$$
	G=\left(\begin{array}{ccccccccccc}
		1&0&0&1&0&1&0&1&1&0&1\\
		0&1&0&1&1&0&0&1&0&1&1\\
		0&0&1&1&1&1&0&0&1&1&1\\
		0&0&0&0&0&0&1&1&1&1&1\\
	\end{array}\right), B=\left(\begin{array}{cccc}
		0&1&0&0\\
		0&0&1&0\\
		0&0&0&1\\
		1&1&0&0\\
	\end{array}\right).$$
	
	Then, by Magma \cite{bosma1997magma}, one can verify that $(G,BG,\ldots, B^{k-1}G)$ generates an $[(11,4),4,11]_2$ Type-II GRC $\mathcal{C}$ with SBDH $(\mathbf{d}_1(\mathcal{C})=5,\mathbf{d}_2(\mathcal{C})=8,\mathbf{d}_3(\mathcal{C})=10,\mathbf{d}_4(\mathcal{C})=11 )$.	 
	By Lemma \ref{Griesmer_Bound}, no $[(11,r),4,\mathbf{d}_4(\mathcal{C})+1]_2$ code exists, so $\mathcal{C}$ is a SBDH-optimal Type-II GRC.
	\hfill $\square$ \end{example}

In \cite{solomon1965algebraically}, Solomon and Stiffler proposed a class of optimal linear codes obtained by puncturing Simplex codes, as shown in the following lemma.

\begin{lemma}[Solomon and Stiffler code, \cite{solomon1965algebraically,Belov1974}]\label{SS_Code}
	Let $k > k_1 > k_2 >  \cdots > k_t \ge 1$ and $s$ be positive integers that satisfy $\sum _{i=1}^{\min \{s+1, t\}} k_{i} \leq s k$. Then, there exists a Griesmer code with parameters 
	\begin{equation}
		\left[sN_k-\sum_{i=1}^{t}N_{k_i}, k, s q^{k-1}-\sum_{i=1}^{t} q^{k_{i}-1}\right]_q.
	\end{equation}
\end{lemma}

According to Theorem \ref{nonCyclic_mound} and Lemma \ref{SS_Code}, we have the following theorem.

\begin{theorem} \label{cor_SS-deduce} 
	Let $k > k_1\ge 1$ and $0\le \Delta \le q-1$ be integers. 
	Let $C$ be an $[\frac{q^{k}-q^{k_1}}{q-1}-\Delta,k,q^{k-1}-q^{k_1-1}-\Delta]_q$ linear code with generator matrix $G$. Then, $\mathcal{G}_k=(G,BG,\ldots, B^{k-1}G)$ generates an SBDH-optimal $\left [(\frac{q^{k}-q^{k_1}}{q-1}-\Delta,k),k \right ]_{q}$ Type-II GRC $\mathcal{C}$ with SBDH  $ 
	\left(\mathbf{d}_r(\mathcal{C})\ge g_q\left(r,q^{k-1}-q^{k_1-1}-\Delta\right) \right)_{r = 1}^k$, where $B\in GL_k(\mathbb{F}_q)$ has no eigenvalue lying in $\mathbb{F}_q$.
\end{theorem}
\begin{proof}
	We denote $d=q^{k-1}-q^{k_1-1}-\Delta$ and $n=\frac{q^{k}-q^{k_1}}{q-1}-\Delta$.
	By Theorem \ref{nonCyclic_mound}, $\mathcal{C}$ has SBDH $ (\mathbf{d}_r(\mathcal{C})\ge g_q(r,d) )_{r = 1}^k $.	
	It is to prove that, for $2\le m\le k$, $[(n,m),k,g(m,d)]_q$ linear block metric codes are  optimal.
	Let $m\le k$ and $D=g(m,d)$. 
	Since $C$ is a Grismer code, we have $n=g(k,d)$. 
	  Now, we prove that $[(n,m),k,g(m,d)]_q$ linear block metric codes are  optimal in two cases.
	  
	  \textbf{Case 1:} If $\Delta=0$,
	we calculate

	\begin{align}
		\begin{array}{rl}
			nN_m-g_{q,m}(k,D)=&\left(\sum\limits_{i=0}^{k-1}
			\left\lceil  \frac{d}{q^i}\right\rceil \right) \left(\sum\limits_{i=0}^{k-1}
			\left\lceil  \frac{q^{m-1}}{q^i}\right\rceil \right)-
			\sum\limits_{i=0}^{k-1}
			\left\lceil  \frac{D}{q^i}\right\rceil				 \\
			=& \sum\limits_{i=0}^{k-m-1} \left(
			q^{m-1} \left\lceil  \frac{d}{q^{m+i}}\right\rceil+ q^{m-2}\left\lceil  \frac{d}{q^{m+i}}\right\rceil +\cdots+ \left\lceil  \frac{d}{q^{m+i}}\right\rceil\right)\\
			&-
				\sum\limits_{i=1}^{k-m} 
				\left\lceil
				 \frac{1}{q^i}\left( d+\left\lceil  \frac{d}{q}\right\rceil+\cdots +\left\lceil  \frac{d}{q^{m-1}}\right\rceil \right)
				 \right\rceil\\	 
		=&\left\{
				\begin{array}{ll}
					\sum\limits_{i=1}^{\min\{ k_1-1,k-k_1\}} 
					\left\lfloor
					\frac{N_{k_1}}{q^i}
					\right\rfloor	,& k_1\le m,\\
					\sum\limits_{i=1}^{\min\{ m-1,k-k_1\}} 
					\left\lfloor
					\frac{N_{m}}{q^i}
					\right\rfloor,	& k_1\ge m+1.
				\end{array}\right.
				\\
		\end{array}
	\end{align}

Given that $\max \left\{ \sum\limits_{i=1}^{\min\{ k_1-1,k-k_1\}} 
\left\lfloor
\frac{N_{k_1}}{q^i}
\right\rfloor, \sum\limits_{i=1}^{\min\{ m-1,k-k_1\}} 
\left\lfloor
\frac{N_{m}}{q^i}
\right\rfloor\right\} =\sum\limits_{i=1}^{\min\{ m-1\}} 
iq^{m-i-1} < N_m$,  and thus we obtain $ nN_m-g_{q,m}(k,D)< N_m$.
Due to 
		\begin{equation}
		\begin{array}{rl}
			g_{q,m}(k,D+1)-g_{q,m}(k,D)=&\sum\limits_{i=0}^{k-1}\left\lceil  \frac{q^{m-1}(D+1)}{q^{i}}\right\rceil-
			\sum\limits_{i=0}^{k-1}\left\lceil  \frac{q^{m-1}D}{q^{i}}\right\rceil\\
			=&\left\{
			\begin{array}{ll}
				N_m+k^{\prime},& q\mid D,\\
				N_m,	& q\nmid D,
			\end{array}\right.
			\\
		\end{array}
	\end{equation}
where $k^{\prime}$ is the maximum integer such that $q^{k^{\prime}}\mid D$,  we get $nN_m< g_{q,m}(k,D+1)$.
Thus, $\mathcal{C}$ is a SBDH-optimal Type-II GRC.

 \textbf{Case 2:} If $1\le \Delta\le q-1$,  based on above results, we have 
 \begin{equation}
 	\begin{array}{rl}
 		nN_m-g_{q,m}(k,D)=& \left\{
 		\begin{array}{ll}
 			\sum\limits_{i=1}^{\min\{ k_1-1,k-k_1\}} 
 			\left\lfloor
 			\frac{N_{k_1}+\Delta}{q^i}
 			\right\rfloor	,& k_1\le m,\\
 			\sum\limits_{i=1}^{\min\{ m-1,k-k_1\}} 
 			\left\lfloor
 			\frac{N_{m}}{q^i}
 			\right\rfloor,	& k_1\ge m+1.
 		\end{array}\right.
 		\\
 	\end{array}
 \end{equation}

For $k_1\ge m+1$, $nN_m-g_{q,m}(k,D)$ has the same value as Case 1, it is  optimal.
For $k_1\le m$, we need to compute the value of $ 			\sum\limits_{i=1}^{\min\{ k_1-1,k-k_1\}} 
\left\lfloor
\frac{N_{k_1}+\Delta}{q^i}
\right\rfloor$. 
If $q=2$, $ \sum\limits_{i=1}^{\min\{ k_1-1,k-k_1\}} 
\left\lfloor
\frac{N_{k_1}+\Delta}{q^i}
\right\rfloor= \sum\limits_{i=1}^{\min\{ k_1-1,k-k_1\}} 2^{i-1}\le N_m-1$.
If $q>2$,  $\sum\limits_{i=1}^{\min\{ k_1-1,k-k_1\}} 
\left\lfloor
\frac{N_{k_1}+\Delta}{q^i}
\right\rfloor= \sum\limits_{i=1}^{\min\{ k_1-1,k-k_1\}} 
\left\lfloor
\frac{N_{k_1}}{q^i}
\right\rfloor + \left\lfloor
\frac{1+\Delta}{q^i}
\right\rfloor < N_m$.
 Thus, for $k_1\le m$, this code is also  optimal.
The calculations confirm that $\mathcal{C}$ is also a SBDH-optimal Type-II GRC.
The theorem follows.
\end{proof}


\begin{example}
	Setting \( q = 2 \), \( k = 5 \), and \( k_1 = 3 \), according to Theorem \ref{cor_SS-deduce}, we obtain a SBDH-optimal Type-II GRC \( \mathcal{C}  \) with parameters \( [(24,5), 5, 24]_2 \). This code features a SBDH  \( (\mathbf{d}_1 = 12, \mathbf{d}_2 = 18, \mathbf{d}_3 = 21, \mathbf{d}_4 = 23, \mathbf{d}_5 = 24 ) \), which corresponds to five  optimal codes with parameters \( [24, 5, 12]_2 \), \( [(24,2), 5, 18]_2 \), \( [(24,3), 5, 21]_2 \), \( [(24,4), 5, 23]_2 \), and \( [(24,5), 5, 24]_2 \), respectively.
In particular, by Lemma \ref{Griesmer_Bound}, there are no linear codes that have parameters $[24,3,18]_4$ and $[24,2,23]_{16}$, so the codes  \( [(24,2), 5, 18]_2 \) and \( [(24,4), 5, 23]_2 \) outperform linear codes. 
As for the \( [(24,5), 5, 24]_2 \) code, the Singleton bound is achieved.
Therefore, the Type-II GRC  \( \mathcal{C}  \) not only performs optimally in each transmission under  block metric, but also outperforms linear codes over $\mathbb{F}_{q^m}$ in error-correction.
\hfill $\square$ \end{example}

\subsection{Construction of Type-II generalized repetition codes from quasi-cyclic codes}

In this subsection, we propose a method to construct good Type-II quasi-cyclic GRCs via computer search.

\begin{theorem} \label{The_lowerbound_QC}
	Suppose $\mathcal{C}$ is a quasi-cyclic code with generator  $\mathbf{g}(x)=(f_1(x)g(x),g(x)f_2(x),\ldots,g(x)f_{m}(x))$, where $1\le m\le n-\deg(g(x))$, $g(x)$ and $f_i(x)$ are polynomials in $\mathbb{R}_{q,n}$, such that $g(x)\mid(x^n-1)$ and $\gcd(f_1(x),\ldots,f_{m}(x),x^n-1)=1$. 
	Let $C$ be a cyclic code of length $n$ generated by $g(x)$.  
	Under the conditions that:
	\begin{enumerate}
		\item $\deg(f_1(x)) < \deg(f_2(x)) < \cdots < \deg(f_{m}(x))$,
		\item $\gcd\left(f_i(x), \frac{x^n - 1}{g(x)}\right) = 1$ for all $i \in [m]$, 
	\end{enumerate}
	the SBDH \((\mathbf{d}_r(\mathcal{C}))_{r = 1}^m\) and SHDH \((\underline{\mathbf{d}}_r(\mathcal{C}))_{r = 1}^m\) of $\mathcal{C}$ satisfy the following lower bounds:
	$$
	\mathbf{d}_r(\mathcal{C}) \geq g_q(r, \mathrm{d}(C)) \quad \text{and} \quad \underline{\mathbf{d}}_r(\mathcal{C}) \geq r \cdot \mathrm{d}(C).
	$$ 
\end{theorem}
\begin{proof}
	Let $\mathbf{c}=(\mathbf{c}_1,\mathbf{c}_2,\ldots,\mathbf{c}_{m})$ be a codeword of $\mathcal{C}$ with the associated polynomial  $\mathbf{c(x)}=a(x)\mathbf{g}(x)=(c_1(x),c_2(x),\ldots,c_m(x))$, where $a(x)\in \mathbb{R}_{q,n}$. 
	Since $\gcd\left(f_i(x), \frac{x^n - 1}{g(x)}\right) = 1$, for all $i \in [m]$, polynomial $c_i(x)=a(x)f_i(x)g(x)$ is associated with the same cyclic code $C$. 
	Thus, we have $\underline{\mathbf{d}}_r(\mathcal{C}) \geq r \cdot \mathrm{d}(C)$. 
	
	Let $g_i^{\prime}(x)=f_i(x)g(x) \pmod{ x^n-1}$. 
	As $\deg(f_1(x)) < \deg(f_2(x)) < \cdots < \deg(f_{m}(x))$ and $1\le m\le n-\deg(g(x))$,  we have that  $g_1^{\prime}(x), g_2^{\prime}(x),\ldots,g_m^{\prime}(x)$ are linearly independent polynomials. 
	Thus, $\mathrm{rank}(\phi_m(\mathbf{c}))=m$ for $\gcd\left(a(x),\frac{x^n-1}{g(x)} \right)=1$. 
	It follows that $\mathbf{d}_r(\mathcal{C}) \geq n(C_1)\ge g_q(r, \mathrm{d}(C))$, where $C_1$ is an $r$-dimensional linear code that generated  by $r$ (of the) vectors among $\mathbf{c}_1, \ldots, \mathbf{c}_m$.

	If $\gcd\left(a(x),\frac{x^n-1}{g(x)} \right)\ne 1$, then $\mathrm{rank}(\phi_m(\mathbf{c}))<m$ and $\phi_m(\mathbf{c})$ defines a cyclic subcode of $C$, which has effective length $n$. 
	Therefore, if the dimension of $C_1$ is less than $r$, then $C_1$ is a cyclic subcode of $C$ and $n(C_1)=n$. 
	Otherwise, we get $\mathbf{d}_r(\mathcal{C}) \geq n(C_1)\ge g_q(r, \mathrm{d}(C))$.
	In summary, we finish the proof. 
\end{proof}

Theorem \ref{The_lowerbound_QC} establishes lower bounds on the $r$-th minimum sub-block and sub-Hamming distances of special quasi-cyclic codes. 
This bound facilitate the search for high-performance Type-II quasi-cyclic GRCs.  
All Type-I regular GRCs derived from Theorems \ref{Theo_SBDH_lower_bounds}–\ref{Theo_QC_Type-I} are bounded by the upper bound stated in Theorem \ref{Theo_MDS_bound}.  
We now provide an example to show that Type-II quasi-cyclic GRCs can surpass the bound in Theorem \ref{Theo_MDS_bound} and attain the Singleton bound described in Lemma \ref{Singleton_bound}.

 
\begin{example}
	Let $F_{11}=\{0,1,\ldots,10\}$ and $n=10$. Then, the $[10,6,5]_{11}$ Reed-Solomon code has generator $g(x)=x^4 + 3x^3 + 5x^2 + 8x + 1$. Let  $f_1(x)=3x^6 + 8x^5 + 4x^4 + x^2 + 7x + 5$, $f_2(x)= 10x^6 + 5x^5 + 7x^4 + 7x^2 + 9x +2$, and $f_3(x)=9x^6 + 4x^5 + 7x^4 + 6x^2 + 6$. 
	Then, $\mathbf{g}(x)=(f_1(x)g(x),f_2(x)g(x),f_3(x)g(x))$ generates a $[30,6,20]_{11}$ quasi-cyclic code $\mathcal{C}$.  
	Since $\gcd(f_i(x),x^{10}-1)=1$ for $i\in [3]$, we have $\mathbf{d}_1(\mathcal{C})=5$. 
 
Moreover, it is easy to use Magma \cite{bosma1997magma} to verify that $\mathcal{C}$ has SBDH and SHDH $(\mathbf{d}_1(\mathcal{C})=5,\mathbf{d}_2(\mathcal{C})=8,\mathbf{d}_3(\mathcal{C})=9)$ and $(\underline{\mathbf{d}}_1(\mathcal{C})=5,\underline{\mathbf{d}}_2(\mathcal{C})=12,\underline{\mathbf{d}}_3(\mathcal{C})=20)$. 
Therefore, the sub-block code of $\mathcal{C}$ for $r=2$ has parameters $[(10,2),6,8]_{11}$, which attains the Singleton bound in Lemma \ref{Singleton_bound}. 
According to the definition in \cite{ball2025griesmer},  $3$-th sub-block code with parameters $[(10,3),6,9]_{11}$ is also a fractional MDS code. 
Thus, $\mathcal{C}$ is (fractional) MDS under block metrics for $r=1,2,3$.  
Moreover, one can verify that $\mathcal{C}$ beats the upper bound in Theorem \ref{Theo_MDS_bound} for both $r=2,3$. 
\hfill $\square$
\end{example}

 In classical coding theory, the primary performance metric for an \([n,k]_q\) linear code \(C\) is its minimum Hamming distance \(\mathrm{d}(C)\). The construction paradigm of Type-II GRCs, however, differs from that of classical linear codes. For an \([(n,m),k]_q\) Type-II GRC \(\mathcal{C}\), the core performance criteria are its SBDH \((\mathbf{d}_r(\mathcal{C}))_{r = 1}^m\) and SHDH \((\underline{\mathbf{d}}_r(\mathcal{C}))_{r = 1}^m\).
 Together, they define a multi-objective optimization problem.  
Guided by this framework, we construct many binary Type-II quasi-cyclic GRCs for $m=2$ with excellent parameters using the algebraic software Magma \cite{bosma1997magma}.  
 Table \ref{Table: Good_Type-II GRC} lists the construction details of these codes, where all generator polynomials are represented as hexadecimal coefficient vectors. 
 For example, the coefficient vector of polynomial $f(x)=x^5 + x^4 + x^3 + x + 1$ is $\mathbf{f}=(1, 1, 0, 1, 1, 1)$. 
 Since $\mathbf{f}$ has length $6$, we need to add two zeros to the beginning of $\mathbf{f}$ to obtain $\mathbf{f}^{\prime}=(0,0,1, 1, 0, 1, 1, 1)$. The hexadecimal value of $\mathbf{f}^{\prime}$ is $(3,7)$.

%
%

\begin{table*}[!]
	\caption{Good Type-II quasi-cyclic GRCs with $m=2$ and generator $\mathbf{g}(x)=(g_1(x),g_2(x))$}
	\label{Table: Good_Type-II GRC}
	\centering
	\resizebox{\textwidth}{!}{
		\begin{threeparttable}
			\begin{tabular}{ccccccc}
				\toprule
				  No.   &   Parameters    &                       $g_1(x)$                       &                       $g_2(x)$                       & $\mathbf{d}_1(\mathcal{C})$ & $\mathbf{d}_2(\mathcal{C})$ & $\underline{\mathbf{d}}_2(\mathcal{C})$ \\ \midrule
				  1    & $[(31,2),6]_2$  &              $ ( 2, 6, 3, C, A, D, D )$              &            $ ( 4, E, 1, A, 9, 1, 7, D )$             &            $15$             &            $23$             &                  $31$                   \\
				  2    & $[(31,2),10]_2$ &               $ ( 3, A, 6, B, 7, F) $                &                     $(6, B, 3)$                      &            $12$             &            $20$             &                  $24$                   \\
				  3    & $[(31,2),11]_2$ &               $ ( 1, 3, 3, C, B, 1) $                &                     $(4, D, 7)$                      &            $11$             &            $19$             &                  $24$                   \\
				  4    & $[(31,2),15]_2$ &                 $ ( 1, B, 1, 4, 9) $                 &                    $(3, E, 9, F)$                    &             $8$             &            $16$             &                  $20$                   \\
				  5    & $[(31,2),20]_2$ &                    $ (A, 2, B) $                     &                 $(1, 8, 1, 0, 8, B)$                 &             $6$             &            $14$             &                  $18$                   \\
				  6    & $[(31,2),21]_2$ &              $ ( 0, 0, 2, C, B, 9, B )$              &            $ ( 0, 9, 2, 0, 2, 1, D, 9 )$             &             $5$             &            $13$             &                  $17$                   \\
				  7    & $[(31,2),25]_2$ &            $ ( 1, B, D, 7, 0, 3, 3, 9 )$             &            $ ( 1, C, F, 5, 8, B, 2, 9 )$             &             $4$             &            $11$             &                  $16$                   \\
				  8    & $[(31,2),26]_2$ &            $ ( 1, 5, 7, C, F, 5, 7, F )$             &            $ ( 1, 3, 0, 2, 7, 3, F, 1 )$             &             $3$             &            $11$             &                  $15$                   \\
				  9    & $[(33,2),10]_2$ &             $( A, 0, A, 9, D, F, B, 9 )$             &             $( 7, 4, 0, 5, C, 9, 6, 9 )$             &            $12$             &            $21$             &                  $26$                   \\
				  10    & $[(33,2),13]_2$ &             $( 0, 9, 6, 5, 5, B, 7, 1 )$             &             $( 0, 8, 0, 8, 3, 9, 2, 7 )$             &            $10$             &            $18$             &                  $24$                   \\
				  11    & $[(33,2),20]_2$ &             $( 1, 6, 6, B, 2, 6, E, 5 )$             &             $( 2, D, 7, 5, E, 8, D, F )$             &             $6$             &            $15$             &                  $18$                   \\
				  12    & $[(33,2),22]_2$ &             $( 8, 5, 9, 0, E, 7, D, 9 )$             &             $( 0, 4, 8, B, C, 3, 5, D )$             &             $6$             &            $14$             &                  $18$                   \\
				  13    & $[(35,2),25]_2$ &           $( 4, 3, B, 2, 4, 8, 6, 7, B )$            &            $( 1, 5, 7, B, 5, A, 1,4, 3 )$            &             $4$             &            $14$             &                  $18$                   \\
				  14    & $[(35,2),27]_2$ &           $( 0, 2, E, 9, 2, 5, 3, 3, 7 )$            &            $( 2, B, C, 2, C, E, 6,A, 3 )$            &             $4$             &            $14$             &                  $16$                   \\
				  15    & $[(35,2),28]_2$ &           $( 3, 4, 9, 6, 4, A, 5, A, D )$            &            $( 3, 2, C, F, 2, 4, 4,1, F )$            &             $4$             &            $13$             &                  $16$                   \\
				  16    & $[(39,2),13]_2$ &          $( 3, C, 8, 0, B, 6, 1, 5, 6, F )$          &          $( 3, 3, 8, F, 3, 8,9, 5, 2, 7 )$           &            $12$             &            $23$             &                  $31$                   \\
				  17    & $[(39,2),24]_2$ &          $( 1, 8, D, D, 3, 3, 6, B, 1, F )$          &          $( 3, 4, 1, 6, A, 7,D, D, 0, 1 )$           &             $6$             &            $17$             &                  $20$                   \\
				  18    & $[(39,2),26]_2$ &          $( 1, 1, 6, 2, 6, 9, 2, A, B, D )$          &          $( 3, 5, B, 5, 6, F,A, B, D, D )$           &             $6$             &            $16$             &                  $20$                   \\
				  19    & $[(41,2),20]_2$ &           $( B, A, 2, D, E, A, 1, D, F )$            &          $( E, 8, 0, 6, 2, 2, 3,6, C, 3 )$           &            $10$             &            $21$             &                  $26$                   \\
				  20    & $[(41,2),21]_2$ &        $( 0, B, 9, 5, 6, A, 2, 3, B, 6, F )$         &          $( E, 3, 7, A, D,6, 0, E, 7, B )$           &             $9$             &            $20$             &                  $25$                   \\
				  21    & $[(43,2),14]_2$ &        $( 2, E, E, 0, 2, 6, 2, E, 5, F, 5 )$         &         $( 1, 5, 4, A, D,C, 5, A, 3, B, 5 )$         &            $14$             &            $27$             &                  $32$                   \\
				  22    & $[(43,2),15]_2$ &        $( 2, 7, 5, F, 9, 9, F, A, E, 5, F )$         &            $( 5, 5, 3, D, E,1, B, 9, B )$            &            $13$             &            $26$             &                  $32$                   \\
				  23    & $[(43,2),28]_2$ &          $( 3, 6, 4, 0, 8, F, F, 4, 2, F )$          &          $( 9, A, 6, 2, C, 5,0, 8, 3, F )$           &             $6$             &            $18$             &                  $22$                   \\
				  24    & $[(43,2),29]_2$ &        $( 7, 3, 8, 9, 0, 8, A, 2, 0, 2, 5 )$         &          $( 3, 9, E, 3, 2,9, 7, 6, 8, 1 )$           &             $6$             &            $17$             &                  $22$                   \\
				  25    & $[(45,2),35]_2$ &        $ ( 7, 5, E, D, 2, B, 1, 0, 7, C, 9 )$        &        $ ( B, F, 8, 6, 7,7, 2, 8, A, D, 5 )$         &             $4$             &            $16$             &                  $20$                   \\
				  26    & $[(47,2),23]_2$ &       $( 2, E, 0, 1, 5, 9, 0, 3, E, 1, 7, 3 )$       &       $( 0, C, 8, D,7, 5, 2, 9, 6, 0, B, 1 )$        &            $12$             &            $24$             &                  $28$                   \\
				  27    & $[(47,2),24]_2$ &       $( 7, 0, D, 2, 4, F, 7, 3, 3, 8, E, F )$       &       $( 3, B, 9, 8,2, 2, 3, F, 9, C, 5, 1 )$        &            $11$             &            $22$             &                  $28$                   \\
				  28    & $[(51,2),16]_2$ &     $ ( 3, D, 8, 6, 9, C, 0, 0, 2, 1, 6, 1, B )$     &     $ ( 5, 7, B,0, 6, 1, 3, 1, 4, 5, 7, 4, B )$      &            $16$             &            $30$             &                  $36$                   \\
				  29    & $[(51,2),17]_2$ &      $ ( 0, E, C, 7, B, 4, 2, 6, 2, 6, 5, D )$       &     $ ( 3, 9, 1, 7,E, 4, C, B, 6, 3, D, F, B )$      &            $16$             &            $28$             &                  $36$                   \\
				  30    & $[(51,2),19]_2$ &      $ ( 9, F, 1, 9, 9, 0, 0, D, 1, 0, 3, F )$       &     $ ( 1, 8, 4, 2,A, 2, 4, 7, 3, 7, 3, 3, 3 )$      &            $14$             &            $27$             &                  $35$                   \\
				  31    & $[(51,2),26]_2$ &     $ ( 1, 0, 6, 3, 9, 8, C, C, F, A, 6, C, 3 )$     &     $ ( 2, C, B,7, 9, 6, C, 6, 2, A, E, A, D )$      &            $10$             &            $22$             &                  $32$                   \\
				  32    & $[(51,2),40]_2$ &     $ ( 0, E, 3, C, 7, B, A, 7, 0, C, F, E, B )$     &     $ ( 0, 5, 2,2, 1, C, F, 9, A, D, E, 0, D )$      &             $4$             &            $17$             &                  $22$                   \\
				  33    & $[(51,2),42]_2$ &     $ ( 0, 8, A, A, 7, 8, 3, 0, 0, 7, 5, 3, 3 )$     &     $ ( 5, 4, 5,B, 5, 0, A, 0, 3, 1, 6, 3, B )$      &             $4$             &            $17$             &                  $20$                   \\
				  34    & $[(55,2),21]_2$ &   $ ( 6, 9, 9, D, 8, 8, E, 9, D, 2, A, 6, 5, 1 )$    &    $ ( 2, 1,F, 8, 6, 3, 2, 4, 5, 3, 3, F, 2, 7 )$    &            $15$             &            $29$             &                  $36$                   \\
				  35    & $[(55,2),24]_2$ &     $ ( 9, D, 2, 1, 4, 1, 8, 9, 0, 3, 3, F, 3 )$     &    $ ( 3, A, 7,5, 3, 2, 6, F, C, A, 5, 1, C, B )$    &            $12$             &            $27$             &                  $36$                   \\
				  36    & $[(55,2),34]_2$ &   $ ( 2, A, 4, F, 0, 9, 2, C, 1, 4, 6, F, A, 1 )$    &    $ ( 6, 0,3, C, 3, 7, 5, A, E, 7, B, F, 1, 7 )$    &             $8$             &            $22$             &                  $28$                   \\
				  37    & $[(63,2),9]_2$  & $( 4, 6, A, E, 3, 9, D, D, 4, 9, 8, 6, 3, 3, 2, 7 )$ &  $( B, 0, A, 3, 0, 0, C, A, C, 8, 7, 6, C, 7, F )$   &            $28$             &            $44$             &                  $56$                   \\
				  38    & $[(63,2),10]_2$ & $( 0, 3, 4, 0, 4, B, C, A, 7, 2, C, F, E, 1, 4, 1 )$ &  $( 0, F, 5, F, 7, C, B, F, B, 4, 3, 5, 8, D, 3 )$   &            $27$             &            $42$             &                  $56$                   \\
				  39    & $[(63,2),11]_2$ & $( 5, 6, 0, 8, F, C, C, E, C, 4, 6, 2, F, B, B, 5 )$ & $( 4, C, C, 2, 6, 8, E, 5, 9, 7, C, A, B, A, 0, 5 )$ &            $26$             &            $42$             &                  $52$                   \\
				  40    & $[(63,2),12]_2$ & $( 0, 6, 9, 5, 6, A, 6, E, 4, 4, 6, 2, 2, 1, D, 7 )$ & $( 2, E, 2, D, E, A, 7, 8, 7, 4, 7, E, B, 4, 3, B )$ &            $24$             &            $42$             &                  $52$                   \\
				  41    & $[(63,2),13]_2$ & $( 0, F, C, 5, 6, B, 9, 8, 0, 2, 0, 0, C, E, B, 5 )$ & $( 0, 3, 1, 4, 4, 7, F, 1, 1, 4, 6, 0, B, 2, 4, D )$ &            $24$             &            $40$             &                  $52$                   \\
				  42    & $[(63,2),15]_2$ & $( 3, 8, 7, 5, 3, 7, 9, 7, 7, 9, 3, 7, C, 5, 7, 1 )$ & $( 5, 2, 7, 9, 0, A, 4, 3, 4, 3, 2, C, A, 2, D, 7 )$ &            $24$             &            $36$             &                  $48$                   \\
				  43    & $[(63,2),16]_2$ & $( 2, 1, E, 9, 4, A, 4, 8, C, 2, 3, C, B, 1, 5, 9 )$ & $( 7, 6, 5, D, E, 3, 1, E, C, 9, 1, 1, 0, 9, E, 1 )$ &            $23$             &            $38$             &                  $48$                   \\
				  44    & $[(63,2),17]_2$ & $( 1, 1, 5, 5, D, 6, 0, E, 7, 8, 6, 6, 9, 8, A, 7 )$ & $( 7, 0, 9, 3, 6, A, 4, F, B, A, 0, D, 7, 2, E, B )$ &            $22$             &            $38$             &                  $48$                   \\
				  45    & $[(63,2),18]_2$ & $( 1, 7, 0, 5, F, 3, 9, F, 5, 7, 4, 2, F, A, 9, F )$ & $( 0, 7, D, C, D, C, F, C, 7, D, D, 5, 0, 4, 2, 3 )$ &            $21$             &            $36$             &                  $47$                   \\
				   46    & $[(63,2),20]_2$ &  $( 5, C, B, D, 8, A, 2, 5, 5, B, 2, 7, A, B, 5 )$   & $( 2,5, 8, 5, B, A, 0, F, 6, 1, 8, B, 1, 3, C, 7 )$  &            $18$             &            $35$             &                  $36$                   \\
				  47    & $[(63,2),21]_2$ & $( 3, F, E, D, 0, 6, 4, B, 3, F, 4, 1, 1, B, C, 9 )$ & $( 2, B, 0, 9, 3, 7, 4, 2, 0, B, 7, 5, 1, 9, 8, 1 )$ &            $18$             &            $34$             &                  $44$                   \\
				  48    & $[(63,2),23]_2$ & $( 3, 6, A, 7, 8, D, D, 1, C, 1, 3, 8, 5, 0, D, 3 )$ & $( 6, A, 3, 1, 5, E, 4, 0, D, 1, C, 7, E, A, 8, 3 )$ &            $16$             &            $33$             &                  $42$                   \\
				  49    & $[(63,2),24]_2$ &  $( 0, 5, 2, 9, 2, E, 8, F, A, 1, 5, 6, 4, 9, 3 )$   & $( 7,3, 5, 2, 5, 6, D, 2, B, 7, 2, 7, 1, C, 2, 3 )$  &            $16$             &            $33$             &                  $40$                   \\
				  50    & $[(63,2),27]_2$ & $( 3, F, 7, 9, 5, 0, 1, 8, C, 5, C, A, 9, 7, 1, 7 )$ & $( 1, B, D, 3, D, 4, A, 5, 1, C, 6, A, C, 7, 6, 1 )$ &            $16$             &            $30$             &                  $40$                   \\
				   51    & $[(63,2),28]_2$ & $( 2, 7, 7, B, F, B, 2, 1, 0, E, 1, C, 4, 8, 8, D )$ & $( 2, 1, 5, 3, 9, B, E, 9, E, 7, 4, 8, 0, 1, 4, 9 )$ &            $15$             &            $30$             &                  $39$                   \\
				  52    & $[(63,2),29]_2$ & $( 0, 0, A, 0, 2, 8, E, 1, 5, 8, 2, A, 2, F, 0, D )$ & $( 6, 1, 4, 2, 1, 6, 6, 4, 7, 2, 6, E, 4, F, C, 1 )$ &            $14$             &            $30$             &                  $38$                   \\
				  53    & $[(63,2),30]_2$ & $( 1, 6, 8, 1, 4, 1, 2, 4, 2, D, 9, 7, D, 3, 6, F )$ & $( 6, 9, 1, 9, 9, C, 7, F, 6, 5, 7, 4, 6, 7, 0, 9 )$ &            $13$             &            $29$             &                  $36$                   \\
				   54    & $[(63,2),31]_2$ & $( 6, C, 1, C, 2, 0, A, 4, C, 0, 4, 0, 3, E, 3, 3 )$ & $( 2, D, 6, D, D, 4, 8, 3, 3, E, 6, F, 9, 3, 5, 9 )$ &            $12$             &            $28$             &                  $36$                   \\
				   55    & $[(63,2),32]_2$ & $( 0, 1, 4, 7, 5, C, 0, A, 4, 9, 7, 2, 2, C, 9, 1 )$ & $( 3, A, C, D, 4, 2, E, E, 8, A, 0, 8, C, 0, F, 1 )$ &            $12$             &            $28$             &                  $34$                   \\
				  56    & $[(63,2),33]_2$ & $( 0, F, F, E, A, 3, F, 7, 3, D, E, 6, 1, D, 3, F )$ & $( 1, 7, 3, 4, F, F, E, 0, 7, 5, 2, 1, A, C, C, B )$ &            $12$             &            $27$             &                  $34$                   \\
				  \bottomrule
			\end{tabular}%
		\end{threeparttable}  		}  
	\end{table*}

	\section{Conclusion}\label{Sect VI}

In this paper, we investigated the error-correction problem in repetitive communication scenarios and introduced the concepts of Type-I and Type-II GRCs, as detailed in Definition \ref{D_repetition code}. We then formulated combinatorial metrics for GRCs. 
Due to their multiple metric properties, GRCs can correct more error patterns compared to classical linear codes, as shown in Table \ref{Comparison of Error Correction Capabilities} and Figure \ref{fig_second_case}. We further demonstrated the superior error-correction capability of GRCs through simulations. 
In Section \ref{Sect IV} and \ref{Sect V}, we established theoretical bounds and presented many explicit construction methods for both Type-I and Type-II GRCs. 
For future research, it would be worthwhile to investigate tighter bounds and develop low-complexity decoding algorithms for Type-I and Type-II GRCs.


	%

	\bibliographystyle{IEEEtran} 
	\bibliography{reference} 

\end{document}